\documentclass[11pt]{article}

\usepackage[T1]{fontenc}
\usepackage[utf8]{inputenc}
\usepackage{newpxtext}
\usepackage{iftex}
\ifPDFTeX
  \usepackage[activate={true,nocompatibility},final]{microtype}
\else
  \usepackage[protrusion=true,final]{microtype}
\fi

\usepackage[margin=1in]{geometry}

\usepackage{authblk}

\usepackage{amsmath}
\usepackage{amssymb,amsfonts}
\usepackage{newpxmath}
\usepackage{mathtools}
\usepackage{bm}
\usepackage{nicefrac}
\allowdisplaybreaks

\usepackage{amsthm}

\usepackage{graphicx}
\graphicspath{{figures/}{./}}
\usepackage{booktabs}
\usepackage{tabularx}
\usepackage{multirow,makecell,array}
\usepackage{float}
\usepackage[font=small,labelfont=bf]{caption}
\usepackage{subcaption}
\usepackage{algorithm}
\usepackage{algpseudocode}

\usepackage{enumitem}
\setlist[itemize]{leftmargin=2.2em,itemsep=2pt,topsep=2pt}
\setlist[enumerate]{leftmargin=2.2em,itemsep=2pt,topsep=2pt}

\usepackage{xcolor}
\definecolor{LinkColor}{rgb}{0.10,0.40,0.75}
\definecolor{CiteColor}{rgb}{0.70,0.25,0.20}
\definecolor{UrlColor} {rgb}{0.20,0.50,0.50}
\definecolor{TodoColor}{rgb}{0.80,0.30,0.10}
\definecolor{HeroBlue}{HTML}{246B89}
\definecolor{HeroRust}{HTML}{B45A3C}
\definecolor{HeroInk}{HTML}{24323D}
\definecolor{HeroMist}{HTML}{F4F7F8}

\usepackage{tikz}
\usetikzlibrary{positioning,calc,arrows.meta,matrix,fit}

\usepackage[numbers,square,sort,comma]{natbib}

\usepackage{url}
\usepackage{hyperref}
\hypersetup{
  colorlinks=true,
  linkcolor=LinkColor,
  citecolor=CiteColor,
  urlcolor=UrlColor,
  breaklinks=true,
  bookmarksnumbered=true,
}
\usepackage{bookmark}

\numberwithin{equation}{section}

\newif\ifdraft \draftfalse
\ifdraft
  \usepackage{lineno}\linenumbers
  \newcommand{\todo}[1]{\textcolor{TodoColor}{\textbf{[TODO:}~#1\textbf{]}}}
\else
  \newcommand{\todo}[1]{}
\fi

\newcommand{\N}{\mathbb{N}}

\newcommand{\E}{\mathbb{E}}
 
\renewcommand{\P}{\mathbb{P}}
\newcommand{\mc}[1]{\mathcal{#1}}

\newcommand{\defeq}{\coloneqq}

\DeclareMathOperator{\poly}{poly}

\DeclarePairedDelimiterX{\inner}[2]{\langle}{\rangle}{#1,#2}

\newcommand{\Cstar}{\mathcal{C}^{*}}
\newcommand{\Lb}{L_{b}}
\newcommand{\tuple}[1]{\langle #1 \rangle}
\DeclareMathOperator{\Cdim}{Cdim}
\DeclareMathOperator{\cl}{cl}
\DeclareMathOperator{\negl}{negl}
\DeclareMathOperator{\mist}{mist}

\theoremstyle{plain}
\newtheorem{theorem}{Theorem}[section]
\newtheorem{proposition}[theorem]{Proposition}
\newtheorem{lemma}[theorem]{Lemma}

\newtheorem{fact}[theorem]{Fact}

\theoremstyle{definition}
\newtheorem{definition}[theorem]{Definition}

\newtheorem{example}[theorem]{Example}

\theoremstyle{remark}
\newtheorem{remark}[theorem]{Remark}

\usepackage{thm-restate}
\usepackage[capitalise,nameinlink,noabbrev,sort&compress]{cleveref}

\crefname{section}{Section}{Sections}             \Crefname{section}{Section}{Sections}
\crefname{subsection}{Section}{Sections}          \Crefname{subsection}{Section}{Sections}
\crefname{equation}{Equation}{Equations}          \Crefname{equation}{Equation}{Equations}
\crefname{figure}{Figure}{Figures}                \Crefname{figure}{Figure}{Figures}
\crefname{table}{Table}{Tables}                   \Crefname{table}{Table}{Tables}
\crefname{theorem}{Theorem}{Theorems}             \Crefname{theorem}{Theorem}{Theorems}
\crefname{proposition}{Proposition}{Propositions}
\Crefname{proposition}{Proposition}{Propositions}
\crefname{lemma}{Lemma}{Lemmas}                   \Crefname{lemma}{Lemma}{Lemmas}
\crefname{corollary}{Corollary}{Corollaries}      \Crefname{corollary}{Corollary}{Corollaries}
\crefname{conjecture}{Conjecture}{Conjectures}    \Crefname{conjecture}{Conjecture}{Conjectures}
\crefname{fact}{Fact}{Facts}                      \Crefname{fact}{Fact}{Facts}
\crefname{definition}{Definition}{Definitions}    \Crefname{definition}{Definition}{Definitions}
\crefname{assumption}{Assumption}{Assumptions}    \Crefname{assumption}{Assumption}{Assumptions}
\crefname{example}{Example}{Examples}             \Crefname{example}{Example}{Examples}
\crefname{problem}{Problem}{Problems}             \Crefname{problem}{Problem}{Problems}
\crefname{remark}{Remark}{Remarks}                \Crefname{remark}{Remark}{Remarks}
\crefname{claim}{Claim}{Claims}                   \Crefname{claim}{Claim}{Claims}
\crefname{algorithm}{Algorithm}{Algorithms}       \Crefname{algorithm}{Algorithm}{Algorithms}

\AddToHook{env/theorem/begin}{\crefalias{section}{theorem}}
\AddToHook{env/proposition/begin}{\crefalias{theorem}{proposition}}
\AddToHook{env/lemma/begin}{\crefalias{theorem}{lemma}}
\AddToHook{env/corollary/begin}{\crefalias{theorem}{corollary}}
\AddToHook{env/conjecture/begin}{\crefalias{theorem}{conjecture}}
\AddToHook{env/fact/begin}{\crefalias{theorem}{fact}}
\AddToHook{env/definition/begin}{\crefalias{theorem}{definition}}
\AddToHook{env/assumption/begin}{\crefalias{theorem}{assumption}}
\AddToHook{env/example/begin}{\crefalias{theorem}{example}}
\AddToHook{env/problem/begin}{\crefalias{theorem}{problem}}
\AddToHook{env/remark/begin}{\crefalias{theorem}{remark}}
\AddToHook{env/claim/begin}{\crefalias{theorem}{claim}}

\title{On Computational Hardness of Mistake-Bounded Language Generation: A Random-Oracle Query Separation}

\newlength{\affiliationinstitutionwidth}
\newlength{\affiliationemailwidth}
\author{%
\makebox[\textwidth][c]{%
  Xiaoyu Li\textsuperscript{1}\qquad
  Andi Han\textsuperscript{2}\qquad
  Dai Shi\textsuperscript{3}\qquad
  Jiaojiao Jiang\textsuperscript{1}\qquad
  Junbin Gao\textsuperscript{2}}\\
\makebox[\textwidth][c]{%
  \parbox[t]{\affiliationinstitutionwidth}{%
    \fontsize{9}{11}\selectfont\normalfont\raggedright
    \makebox[\linewidth][l]{\strut\textsuperscript{1}University of New South Wales}\par
    \makebox[\linewidth][l]{\strut\textsuperscript{2}University of Sydney}\par
    \makebox[\linewidth][l]{\strut\textsuperscript{3}University of Cambridge}\par}%
  \hspace{0.30in}%
  \parbox[t]{\affiliationemailwidth}{%
    \fontsize{9.4}{11}\selectfont\ttfamily\raggedright
    \makebox[\linewidth][l]{\strut\{xiaoyu.li2,jiaojiao.jiang\}@unsw.edu.au}\par
    \makebox[\linewidth][l]{\strut\{andi.han,junbin.gao\}@sydney.edu.au}\par
    \makebox[\linewidth][l]{\strut ds2213@cam.ac.uk}\par}}%
}
\date{}

\hypersetup{
  pdftitle={On Computational Hardness of Mistake-Bounded Language Generation: A Random-Oracle Query Separation},
  pdfauthor={Xiaoyu Li, Andi Han, Dai Shi, Jiaojiao Jiang, Junbin Gao}
}

\begin{document}
\maketitle
\vspace{-3.5em}

\begin{abstract}
Generation in the limit guarantees eventual generation for every countable collection of infinite languages in the model of Kleinberg and Mullainathan~\citep{kleinberg2024generation}, while closure dimension characterizes stronger information-theoretic guarantees~\citep{liramantewari2025generation}. Neither restricts per-output computation.  The cumulative-mistake objective in mistake-bounded generation makes finite failure prefixes quantitative~\citep{kleinbergpealereingold2026mistake}, and a per-output query budget exposes their computational source.  Polynomial-time algorithms are known for parities, conjunctions, and monotone functions with polynomially many maxterms~\citep{jimenezetal2026polytime}. We ask whether information-theoretic ease can coexist with bounded-access computational hardness. Relative to a random oracle $H$, we answer yes by constructing a countable collection $\Cstar$ of infinite languages with $\Cdim(\Cstar)=0$.  Almost surely on the same $H$, an unbounded generator makes zero mistakes on every target and every complete distinct enumeration.  Yet, writing $\lambda$ for the target-seed length, every fixed uniform generator $G$ with $\poly(\lambda,i)$ oracle queries at output $i$ has a constant $c_G>0$ such that, for every sufficiently large $\lambda$, some target incurs more than $2^{c_G\lambda}$ expected mistakes within its first $2(\lceil2^{c_G\lambda}\rceil+1)$ canonical outputs.  Infinite accidental agreement enables exhaustive search; sparse queries hide fresh target values.  Thus, in the random-oracle model, zero-mistake information-theoretic generation coexists with a generator-dependent exponential lower bound on worst-case expected mistakes under polynomial-query access.

\end{abstract}

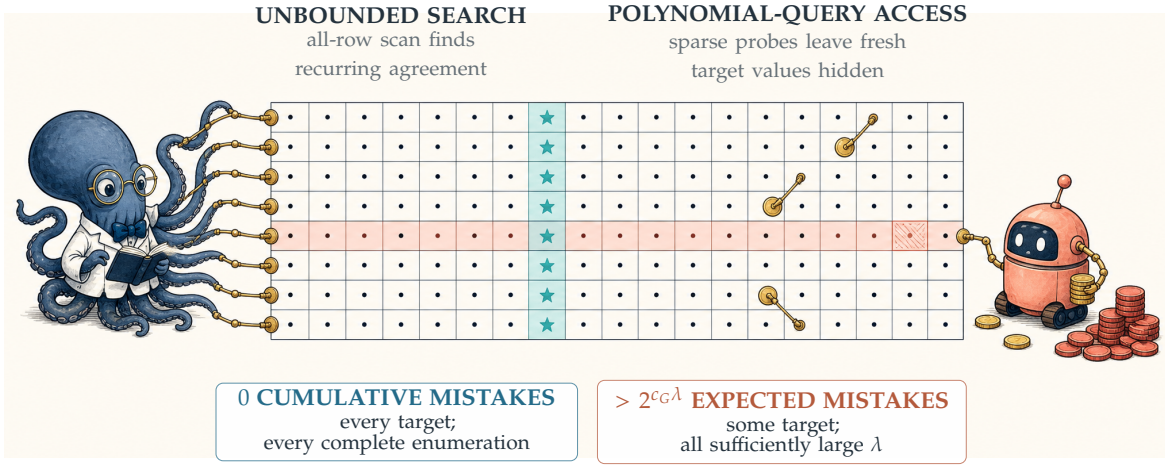
\begin{figure}[H]
\centering
\vspace{.25em}
\resizebox{0.94\linewidth}{!}{%
\begin{tikzpicture}[
  x=1cm,
  y=1cm,
  easybadge/.style={
    rounded corners=2.8pt,
    draw=HeroBlue!72,
    fill=white,
    fill opacity=.95,
    text opacity=1,
    text width=4.80cm,
    minimum height=.88cm,
    inner sep=2.5pt,
    align=center,
    font=\scriptsize,
    text=HeroInk
  },
  hardbadge/.style={
    rounded corners=2.8pt,
    draw=HeroRust!72,
    fill=white,
    fill opacity=.95,
    text opacity=1,
    text width=4.90cm,
    minimum height=.88cm,
    inner sep=2.5pt,
    align=center,
    font=\scriptsize,
    text=HeroInk
  }
]
\path[use as bounding box] (0,0) rectangle (16.05,6.10);
\node[anchor=south west, inner sep=0pt] at (0,0)
  {\includegraphics[width=16.05cm]{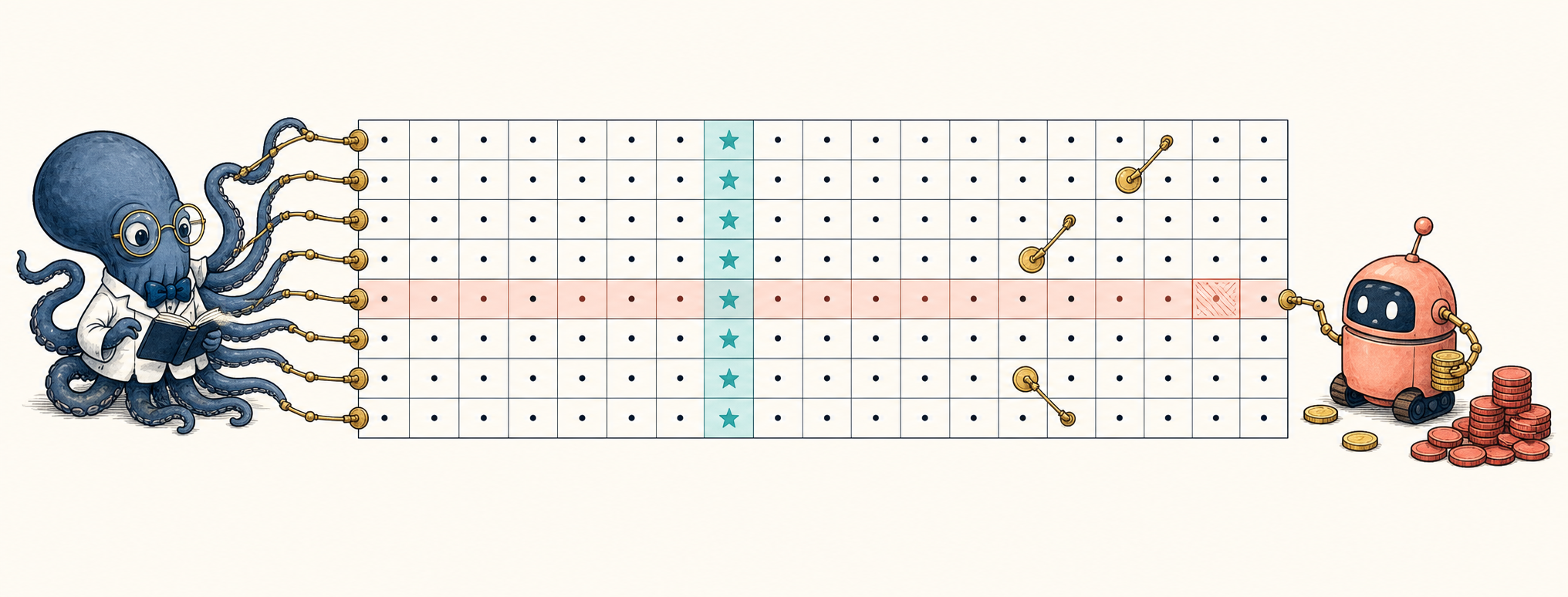}};

\node[font=\footnotesize, text=HeroInk, text width=4.65cm, align=center]
  at (5.32,5.70) {%
    \textbf{UNBOUNDED SEARCH}\\[-1pt]
    {\scriptsize\color{HeroInk!70} all-row scan finds\\recurring agreement}%
  };

\node[font=\footnotesize, text=HeroInk, text width=5.05cm, align=center]
  at (10.78,5.70) {%
    \textbf{POLYNOMIAL-QUERY ACCESS}\\[-1pt]
    {\scriptsize\color{HeroInk!70} sparse probes leave fresh target values hidden}%
  };

\node[easybadge] at (5.40,.50) {%
  {\footnotesize\bfseries\color{HeroBlue} $0$ CUMULATIVE MISTAKES}\\[-1pt]
  every target;\\[-2pt]
  every complete enumeration%
};

\node[hardbadge] at (10.70,.50) {%
  {\footnotesize\bfseries\color{HeroRust} $>2^{c_G\lambda}$ EXPECTED MISTAKES}\\[-1pt]
  some target;\\[-2pt]
  all sufficiently large $\lambda$%
};
\end{tikzpicture}
}

\captionsetup{font=footnotesize}
\caption{\textbf{One oracle graph, two access regimes.}
Unbounded search finds recurring all-row agreement; polynomial-query access can miss fresh target values.
The resulting zero-versus-exponential mistake separation is formalized in Theorem~\ref{thm:main}.}
\label[figure]{fig:separation-at-a-glance}
\end{figure}

\clearpage
\tableofcontents
\clearpage
\section{Introduction}
\label[section]{sec:intro}

Language generation has two logically distinct tasks: valid unseen outputs must exist, and a generator must be able to find them.  For countable collections, the first task is sharply understood at the information-theoretic level.  In the generation-in-the-limit model of Kleinberg and Mullainathan~\citep{kleinberg2024generation}, every countable collection of infinite languages is generatable, even though classical identification from positive data can be impossible~\citep{gold1967language,angluin1980inductive}; closure dimension characterizes stronger generation guarantees~\citep{liramantewari2025generation}.  These results deliberately leave the search for the next valid output computationally unrestricted.

That freedom matters.  A generator that errs for its first $K$ rounds and is thereafter always valid succeeds in generation in the limit whether $K=10$ or $K=10^{10^{10}}$.  The cumulative-mistake objective records $\min\{K,W\}$ failures by horizon $W$~\citep{kleinbergpealereingold2026mistake}. Mistakes are not units of computation, but under a per-output query budget an unsuccessful search for a fresh target point becomes visible as a mistake.  Thus the objective exposes computational burn-in without asserting eventual failure.  Polynomial-time algorithms are known for several natural finite-language families~\citep{jimenezetal2026polytime}.  What is missing in the feedback-free cumulative-mistake setting is a separation between unrestricted generation and bounded per-output access.

We ask whether information-theoretic ease can coexist with strong query hardness: can an unbounded generator be perfect while every query-efficient generator incurs exponentially many mistakes?  Relative to a random oracle, the answer is yes.

\paragraph{One object, two consequences.} Fix a random oracle $H$.  For a seed $b\in\{0,1\}^{\lambda}$, write
\[
  E(b,i)=\tuple{\lambda,i,H(b,i)},
  \qquad
  L_b=\bigl\{E(b,i):i\in\N\bigr\},
\]
and let $\Cstar$ contain these graph languages over all lengths and seeds.  A finite history leaves finitely many consistent seeds, which accidentally share a value at infinitely many fresh indices; for two seeds, each such agreement has probability $2^{-\lambda}$.  An unbounded closure search can use this infinite tail.  A bounded-query generator, however, cannot identify a uniformly hidden seed, so its value at a fresh index is essentially a blind $\lambda$-bit guess.

\paragraph{Main result.} In the security-parameterized model of \cref{sec:prelim}, almost surely over a random oracle $H$, both of the following hold for $\Cstar$:
\begin{itemize}[leftmargin=1.8em]
\item \textbf{Information-theoretic ease.}\enspace $\Cdim(\Cstar)=0$, and a computationally unbounded generator makes \emph{zero} mistakes on every target $\Lb$ under every complete distinct enumeration.
\item \textbf{Computational hardness.}\enspace Every uniform polynomial-query oracle generator $G$ has $c_G>0$ such that, at every sufficiently large length $\lambda$, some length-$\lambda$ target forces $>2^{c_G\lambda}$ expected mistakes (over $G$'s coins) within a generator-dependent exponential horizon under its canonical enumeration.
\end{itemize}

The exact horizon and quantifiers appear in \cref{thm:as-exp,thm:main}; \cref{sec:beyond-canonical} extends the lower bound to target--oracle-independent precommitted schedules. The all-row agreement used at the cold start removes only the possible first mistake; the exponential lower bound comes instead from the difficulty of predicting a fresh target value.

\subsection{Technique overview}
\label[subsection]{subsec:technique-overview}

\paragraph{Unbounded search on the common graph.} At the cold start the unbounded generator knows only $\lambda$, so it scans all $2^\lambda$ length-$\lambda$ rows until they share a value.  After examples arrive, it scans the common graph of the finitely many consistent seeds for an unrevealed point.  Both searches terminate almost surely: the first uses lengthwise agreement, the second the infinite-closure property~(S2).  Every output is valid, although the search may be arbitrarily expensive.

\paragraph{Planting neutralizes adaptive seed tests.} Given samples $\{E(b,i):i\in Q\}$ for a finite $Q$ independent of $(b,H)$, plant their values at the hidden row $b$ of an otherwise independent oracle.  In the resulting shadow execution, a $T$-query algorithm tests a $b$-independent set of at most $T$ seeds, hitting the true row with probability at most $T2^{-\lambda}$.  Without such a hit, a fresh target value remains a uniform $\lambda$-bit guess, adding at most $2^{-\lambda}$.  This coupling proves \cref{lem:s3-formal} even for adaptive algorithms.

\paragraph{A dense band of correct outputs yields a predictor.} If $G$ makes at most $m$ expected mistakes among its first $W=2(m+1)$ outputs, a uniformly random output in that window is a non-mistake with probability at least $1/2$.  By the canonical graph property, stopping there produces a fresh-point predictor.  If $q_G(\lambda,i)\le r_G(\lambda)(i+1)^d$, then with
\[
  m_c=\lceil2^{c\lambda}\rceil,
  \qquad
  W_c=2(m_c+1),
\]
the stopped simulation uses at most $W_c r_G(\lambda)(W_c+1)^d=2^{c(d+1)\lambda+o(\lambda)}$ queries.  For every $0<c<1/(d+1)$, the planted bound therefore makes the low-mistake seed--oracle pairs $2^{-\Omega(\lambda)}$-rare.

\paragraph{One oracle works against every uniform generator.} For fixed $G$, let $f_\lambda(H)$ be the fraction of low-mistake seeds.  The joint estimate bounds $\E_H[f_\lambda]$ exponentially; Markov's inequality and the first Borel--Cantelli lemma then give $f_\lambda(H)<\lambda^{-2}<1$ at every sufficiently large length for almost every $H$.  Countability of uniform generators makes this simultaneous, and intersecting with the easiness event puts both sides on one oracle.

\subsection{Computational context}

The nearest models are the feedback-free cumulative-mistake objective \citep{kleinbergpealereingold2026mistake} and its finite length-$n$ computational specialization \citep{jimenezetal2026polytime}.  We retain the score, but index infinite languages by a security parameter and charge random-oracle queries per output while leaving non-oracle running time unrestricted.  The theorem therefore isolates the bounded-access axis: a query separation in an adapted infinite-language model.

Other resource-sensitive results concern sample-complexity barriers for formal-language classes \citep{arenasetal2025complexity}, membership-query limitations \citep{charikarpabbaraju2025facets}, and computability obstructions for representative generation \citep{pealeramanreingold2025representative}.  Limit-learning models may instead enrich the observations with time-bound information \citep{papazovflammarion2025learning} or succinct machine-independent traces \citep{charikarkleinbergpabbaraju2026succinct}.  Those models change the learner's information; ours charges per-output random-oracle queries and scores cumulative mistakes.  For a thesis-length treatment of generation in the limit, see \citep{mehrotra2026learning}.

Along the objective axis, consistency and full breadth can conflict \citep{kalavasismehrotravelegkas2025limits}, while density quantitatively relaxes their tradeoff \citep{kleinbergwei2025density}.  The deterministic $1/2$ optimum extends to partial enumeration, scaled by the revealed subset's density \citep{kleinbergwei2026partial}.

On this density objective, recent work gives a simpler optimal deterministic algorithm, raises the optimal randomized guarantee to $1-1/e$ against oblivious adversaries, and shows that the corresponding optima hold simultaneously for every order in any finite collection \citep{caietal2026dense}.  Other objectives permit infinitely many mistakes when their frequency vanishes \citep{straussbutoicotterell2026hallucinations}, while asymptotic hallucination rates yield a hierarchy beyond finite mistake bounds~\citep{panigrahiweizhang2026hallucination}.  The closest comparisons appear in \cref{tab:comparisons}.

\begin{table}[H]
\centering
\small
\caption{Adjacent results organized by objective and resource perspective.}
\label[table]{tab:comparisons}
{\renewcommand{\arraystretch}{1.07}
\begin{tabularx}{\textwidth}{@{}
  >{\raggedright\arraybackslash}p{.16\textwidth}
  >{\raggedright\arraybackslash}p{.29\textwidth}
  >{\raggedright\arraybackslash}X@{}}
\toprule
Reference & Objective/model & Resource perspective or result \\
\midrule
\citep{arenasetal2025complexity} & generation in the limit for formal-language classes & sample complexity versus description size; triple-exponential regular-language and noncomputable context-free barriers \\
\citep{kleinbergpealereingold2026mistake} & mistake-bounded generation & information-theoretic mistake bounds \\
\citep{jimenezetal2026polytime} & finite length-$n$ mistake-bounded generation & polynomial-time algorithms for natural finite languages \\
\citep{flammarionetal2026spaceefficient,kleinbergetal2026boundedmemory} & generation in the limit & space and memory restrictions, lower bounds, and limitations under different objectives \\
\citep{hannekeetal2026feedback} & generation with mistake or query feedback & feedback access and its power \\
\citep{bhattamishra2025nsphardness} & next-symbol prediction of regular languages & standard-model cryptographic hardness for prediction \\
\midrule
\textbf{This paper} & mistake-bounded generation, no feedback & per-output random-oracle queries; zero-versus-exponential mistake separation \\
\bottomrule
\end{tabularx}
}
\end{table}

\paragraph{Scope and open directions.} The random-oracle model isolates the black-box barrier but does not guarantee a standard-model instantiation~\citep{canettigoldreichhalevi2004random}.  Finite-block pseudorandom error-correcting codes provide relevant ingredients~\citep{christgunn2024pseudorandom}, but the required infinite-index object and hardness for natural language classes remain open.

\paragraph{Organization.} \Cref{sec:prelim} defines the model, illustrates the cumulative-mistake objective, introduces the random-oracle graph collection, and records the probability tools used later. \Cref{sec:results} proves the closure, prediction, finite-horizon, and almost-sure statements. \Cref{sec:beyond-canonical} extends the lower bound to target--oracle-independent precommitted schedules.  \Cref{sec:discussion} develops the computational implications and the routes toward standard-model and natural-class separations.

\section{Model and preliminaries}
\label[section]{sec:prelim}

We now formalize the model used throughout: the feedback-free cumulative-mistake score is applied to an infinite-language collection indexed by a security parameter, and computational access is measured by per-output oracle queries.

\subsection{Security-parameterized mistake-bounded generation}

\begin{definition}[Generation game]\label[definition]{def:model}
For every $\lambda\in\N$, let $\mc L_\lambda$ be a collection of infinite languages over a countable universe $\mc X_\lambda$.  A single generator $G$ interacts with an adversary that fixes a target $L^*\in\mc L_\lambda$ and a complete, distinct enumeration $\sigma=(x_1,x_2,\ldots)$ of $L^*$:
\begin{enumerate}
\item $G$ receives $1^\lambda$ and outputs $a_0=G(1^\lambda,\varnothing)$;
\item at round $t\ge1$, the adversary reveals $x_t$, after which $G$ outputs $a_t=G(1^\lambda,x_1,\ldots,x_t)$.
\end{enumerate}
We view $G$ as one interactive procedure: it may retain state across rounds and uses one private random tape for the entire interaction.  On every oracle on which a procedure is asserted to be a generator, each requested output computation must halt on every valid game history. The generator receives no indication of whether any output is valid.  It may emit any string; an invalid or already revealed output counts as a mistake.  We count
\[
  \mist(G,L^*,\sigma)
  \defeq \sum_{t\ge0}\mathbf 1\!\left[a_t\notin L^*\setminus\{x_1,\ldots,x_t\}\right],
\]
where the revealed set is empty at $t=0$.  For a horizon $W$, write $\mist_{<W}$ for the same sum over $0\le t<W$.  If $G$ is randomized, bounds refer to expectation over its internal coins.  Our lower bound uses the fixed canonical enumeration, so it also holds when the adversary is allowed to adapt.
\end{definition}

Thus mistakes are invalid or stale outputs under the displayed score.  We do not separately penalize repetition of the generator's own prior outputs unless that output has subsequently been revealed.

\begin{definition}[Uniform oracle generator]\label[definition]{def:eff}
A \emph{uniform oracle generator} is a single probabilistic oracle Turing machine, independent of $\lambda$, the target, and the realized oracle.  Each requested output computation halts on every valid history, for every oracle realization and every internal random tape.
\end{definition}

\begin{definition}[Polynomial-query access]\label[definition]{def:query-cost}
A uniform oracle generator $G$ is \emph{polynomial-query} if there is a fixed polynomial $q_G$ such that, while computing its $i$th output, it makes at most $q_G(\lambda,i)$ queries to $H$.  The cap holds on every execution, uniformly over valid histories, oracle realizations, and internal coins; it is not an expected bound.
\end{definition}

We measure oracle queries per output rather than cumulatively and impose no bound on non-oracle running time.  Thus every polynomial-time generator is polynomial-query.

Here ``uniform'' means one oracle Turing machine across all lengths and targets, not the target-uniform rates of uniform generation~\citep{liramantewari2025generation}; within the latter non-uniform regime, Pareto-optimality captures tradeoffs among language-dependent generation times \citep{charikarpabbaraju2025pareto}.  Target-aware access to the seed is outside our model; $P/\poly$-style advice is instead shared by all targets of one length and is not analyzed.  Finally, unlike mistake-bound online learning, the generator receives no correctness feedback; related computability separations study a different prediction game \citep{hasratibendavid2023computable}.

\begin{example}[A finite but hidden burn-in]\label[example]{ex:finite-burn-in}
Fix one realized run and let $M_t$ be its mistake indicator at round $t$.  If, for some $K\ge0$,
\[
  M_t=
  \begin{cases}
    1,&0\le t<K,\\
    0,&t\ge K.
  \end{cases}
  \qquad\text{Then}\qquad
  \mist_{<W}=\min\{K,W\},
  \quad
  \mist=K.
\]
Generation in the limit records only that $K<\infty$, whereas the finite-horizon score records how many failures have occurred by round $W$.  Because the game is feedback-free, the transition point $K$ is an ex post property of the run, not information revealed to the generator.
\end{example}

Mistakes and queries remain distinct coordinates: runs with the same $K$ may use different numbers of queries, while \cref{alg:zero-mistake} has $K=0$ despite potentially unbounded query cost.  The lower bound links these coordinates under polynomial-query access; it does not identify a mistake with a unit of computation.

\subsection{Closure dimension}

\begin{definition}[Closure and closure dimension {\citep[Definition~3.1]{liramantewari2025generation}}]\label[definition]{def:closure-dimension}
For a collection $\mc H$ over a countable universe and a finite set $F$, let $\mc V(F)=\{L\in\mc H:F\subseteq L\}$ be its version space and define
\[
  \cl(F)=
  \begin{cases}
    \displaystyle\bigcap_{L\in\mc V(F)}L,&\mc V(F)\ne\varnothing,\\[2mm]
    \bot,&\mc V(F)=\varnothing.
  \end{cases}
\]
The \emph{closure dimension} $\Cdim(\mc H)$ is the largest $d$ for which some $d$ distinct examples have a finite, nonempty closure; it is zero if no such $d\ge1$ exists and infinite if no largest $d$ exists.
\end{definition}

We use only the definition of closure dimension from that work, not its closure characterization or closure-generator bound: \cref{lem:construction} proves $\Cdim(\Cstar)=0$ directly, and \cref{lem:easy} directly constructs the zero-mistake generator.

\subsection{Random oracle and graph collection}

\begin{definition}[Random-function graph collection]\label[definition]{def:object}
Let $\N=\{1,2,\ldots\}$.  Let $H$ be a uniformly random function that maps each input $(b,i)$, with $b\in\{0,1\}^{\lambda}$ and $i\in\N$, to a uniform value in $\{0,1\}^{\lambda}$.  All parties have oracle access to $H$.  For $b\in\{0,1\}^{\lambda}$ define
\[
  E(b,i)=\tuple{\lambda,i,H(b,i)},\qquad
  \Lb=\{E(b,i):i\in\N\},\qquad
  \Cstar=\bigcup_{\lambda\ge1}\{\Lb:b\in\{0,1\}^{\lambda}\}.
\]
The tuple encoding is self-delimiting and $i$ is written in binary.  Consequently every element carries its length tag, different lengths never collide, and exactly one string with fields $(\lambda,i,\cdot)$ belongs to a fixed $\Lb$.  Given $b$ and a candidate $x$, membership in $\Lb$ uses one oracle query and $\poly(|b|+|x|)$ time.  Each $\Lb$ is infinite. For the game of \cref{def:model}, we instantiate $\mc L_\lambda=\{\Lb:b\in\{0,1\}^{\lambda}\}$ and take $\mc X_\lambda$ to be the countable universe of well-formed tag-$\lambda$ tuples.
\end{definition}

\begin{example}[A finite oracle window]\label[example]{ex:oracle-window}
Take $\lambda=2$, so the four seeds index all four length-$2$ languages.  The following is one possible finite restriction of the oracle; each entry is the raw value $H(b,i)$, while the corresponding language element is the tagged point $\tuple{2,i,H(b,i)}$.

\begin{center}
\begin{tikzpicture}[
  tablecell/.style={draw=black!28, anchor=center, minimum width=1.16cm, minimum height=.48cm,
    inner sep=1pt, font=\small},
  observed/.style={fill=CiteColor!12, draw=CiteColor!75, line width=.7pt},
  closurecell/.style={fill=LinkColor!10, draw=LinkColor!75, dashed, line width=.7pt},
  universal/.style={fill=black!10, draw=black!70, line width=.9pt}
]
\matrix (oracle) [
  matrix of math nodes,
  ampersand replacement=\&,
  nodes={tablecell},
  row sep=-\pgflinewidth,
  column sep=-\pgflinewidth,
  row 1/.style={nodes={fill=black!6, font=\small\bfseries}},
  column 1/.style={nodes={fill=black!6, font=\small\bfseries}}
] {
  b\backslash i \& 1 \& 2 \& 3 \& 4 \& \cdots \\
  00 \& |[observed]|10 \& |[closurecell]|01 \& |[universal]|11 \& 00 \& \cdots \\
  01 \& |[observed]|10 \& |[closurecell]|01 \& |[universal]|11 \& 10 \& \cdots \\
  10 \& 00 \& 10 \& |[universal]|11 \& 01 \& \cdots \\
  11 \& 11 \& 00 \& |[universal]|11 \& 10 \& \cdots \\
};

\end{tikzpicture}

\vspace{0.35em}

\begin{tikzpicture}[
  observed/.style={fill=CiteColor!12, draw=CiteColor!75, line width=.7pt},
  closurecell/.style={fill=LinkColor!10, draw=LinkColor!75, dashed, line width=.7pt},
  universal/.style={fill=black!10, draw=black!70, line width=.9pt}
]
\node[observed, minimum width=.34cm, minimum height=.24cm] (legobs) {};
\node[right=1mm of legobs, draw=none, font=\scriptsize] (legobstext) {observed point $F$};
\node[closurecell, minimum width=.34cm, minimum height=.24cm, right=7mm of legobstext]
  (legclosure) {};
\node[right=1mm of legclosure, draw=none, font=\scriptsize] (legclosuretext)
  {fresh point in $\cl(F)$};
\node[universal, minimum width=.34cm, minimum height=.24cm, right=7mm of legclosuretext]
  (legall) {};
\node[right=1mm of legall, draw=none, font=\scriptsize] {valid for all targets};
\end{tikzpicture}
\end{center}

The observation $F=\{\tuple{2,1,10}\}$ leaves exactly $\mc V(F)=\{L_{00},L_{01}\}$, and their next agreement gives the fresh closure point $\tuple{2,2,01}\in\cl(F)$.  By contrast, all four rows agree at index $3$, so $\tuple{2,3,11}\in\bigcap_{b\in\{0,1\}^{2}}L_b$ is valid even at the empty history.

This finite window only illustrates the two agreement mechanisms.  For $k$ fixed consistent seeds, agreement at any fresh index has probability $2^{-(k-1)\lambda}$; \cref{lem:construction} uses the infinitely many independent indices to prove that such agreements recur infinitely often almost surely.
\end{example}

\subsection{Standard probability tools}
\label[subsection]{sec:probability-tools}

For events $E_1,E_2,\ldots$, write $\limsup_n E_n=\bigcap_{N\ge1}\bigcup_{n\ge N}E_n$ for the event that infinitely many $E_n$ occur.

The following standard facts hold on arbitrary probability spaces.

\begin{fact}[Markov's inequality]\label[fact]{fact:markov}
If $X\ge0$ is measurable and $a>0$, then $\P[X\ge a]\le \E[X]/a$.
\end{fact}

\begin{fact}[Tonelli's theorem]\label[fact]{fact:tonelli}
If $Z\ge0$ is measurable on a product probability space, then
\[
  \E_{\omega_1,\omega_2}[Z]
  =\E_{\omega_1}\!\left[\E_{\omega_2}[Z\mid\omega_1]\right]
  =\E_{\omega_2}\!\left[\E_{\omega_1}[Z\mid\omega_2]\right],
\]
with the common value allowed to be $+\infty$.
\end{fact}

\begin{fact}[First Borel--Cantelli lemma]\label[fact]{fact:first-borel-cantelli}
If $\sum_n\P[E_n]<\infty$, then $\P[\limsup_n E_n]=0$; no independence assumption is needed.
\end{fact}

\begin{fact}[Second Borel--Cantelli lemma]\label[fact]{fact:second-borel-cantelli}
If the events $(E_n)_n$ are mutually independent and $\sum_n\P[E_n]=\infty$, then $\P[\limsup_n E_n]=1$.
\end{fact}

\begin{fact}[Countable full-measure intersection]\label[fact]{fact:countable-intersection}
If $\P[F_k]=1$ for every $k\in\N$, then $\P[\bigcap_k F_k]=1$.  Equivalently, by countable subadditivity, a countable union of null events is null.
\end{fact}

These standard facts are used throughout the main argument and its later variants.  The construction-specific fresh-point prediction bounds and finite-horizon reductions, by contrast, are proved directly in \cref{sec:results,sec:beyond-canonical}.

\section{The separation}
\label[section]{sec:results}

Recall the object $\Cstar$ of \cref{def:object}, formed from the graph of a random oracle.  The proof has three stages.  In \cref{lem:construction}, property~(S1) is syntactic, recurring agreement gives (S2), and the planted-oracle argument gives (S3).  \Cref{lem:easy,lem:hard} then derive information-theoretic easiness and computational hardness.  Finally, the almost-sure argument places both conclusions on one realized oracle and yields the separation of \cref{thm:main}.

\pagebreak[3]
\begin{lemma}[Construction: determined yet hidden]
\label[lemma]{lem:construction}
Relative to a random oracle $H$, the collection $\Cstar$ of \cref{def:object} satisfies:
\begin{enumerate}
\item[\textup{(S1)}] \textbf{\textup{(Canonical graph.)}}\enspace For each $b,i$, $E(b,i)$ is the unique element of index $i$ in any $\Lb$ with $|b|=\lambda$; membership ``$x\in\Lb$?'' given $b$ costs one query to $H$ and $\poly(|b|+|x|)$ time; each $\Lb$ is infinite.
\item[\textup{(S2)}] \textbf{\textup{(Infinite closures.)}}\enspace $\Cdim(\Cstar)=0$ almost surely.
\item[\textup{(S3)}] \textbf{\textup{(The seed is hidden.)}}\enspace For every possibly randomized algorithm $A$ making $T$ queries to $H$, given $\{E(b,i):i\in Q\}$ for a uniformly random hidden seed $b\in\{0,1\}^{\lambda}$ and any finite $Q\subseteq\N$ chosen independently of $(b,H)$,
  \[
    \P_{b,H,\rho_A}\bigl[A\text{ outputs }E(b,j)\text{ for some }j\notin Q\bigr]\;\le\;(T+2)\,2^{-\lambda},
  \]
where $\rho_A$ denotes $A$'s private coins. In particular this is $\negl(\lambda)$ for $T=\poly(\lambda)$ and stays $\le 2^{-\Omega(\lambda)}$ for every $T=2^{o(\lambda)}$.
\end{enumerate}
\end{lemma}

\begin{proof}
\textbf{(S1) Canonical graph.} The tag $\lambda$ and the index $i$ are literal fields of $E(b,i)=\tuple{\lambda,i,H(b,i)}$. Given $b$, deciding $\tuple{\lambda',i,y}\in\Lb$ amounts to checking $\lambda'=|b|$, $|y|=\lambda$, and $H(b,i)=y$: one query to $H$ and $\poly(|b|+|x|)$ time. Distinct indices give distinct elements, so each $\Lb$ is infinite. The self-delimiting tag excludes cross-length collisions, so for each pair $(\lambda,i)$ exactly one string $\tuple{\lambda,i,\cdot}$ lies in any $\Lb$ with $|b|=\lambda$, namely $E(b,i)$.

\smallskip
\textbf{(S2) Closure dimension zero.} We first isolate a probability-one event without conditioning on an oracle-dependent version space.  Fix $\lambda$ and a deterministic set $B_0\subseteq\{0,1\}^{\lambda}$ with $|B_0|\ge2$, and define
\[
  J_{B_0}\;\defeq\;\{\,i\in\N:H(b,i)\text{ is equal for all }b\in B_0\,\}.
\]
For each $i$, the values $\{H(b,i)\}_{b\in B_0}$ are $|B_0|$ independent uniform $\lambda$-bit strings, so
\[
  \P_H[i\in J_{B_0}]=2^{-(|B_0|-1)\lambda}>0.
\]
The events $\{i\in J_{B_0}\}_{i\in\N}$ are independent because they use disjoint oracle coordinates. The second Borel--Cantelli lemma (\cref{fact:second-borel-cantelli}) therefore gives $|J_{B_0}|=\infty$ almost surely.  There are only countably many pairs $(\lambda,B_0)$, so the countable full-measure intersection (\cref{fact:countable-intersection}) yields one event $\mc E_{\mathrm{agr}}$ on which every such fixed $B_0$ has infinitely many agreement indices.

Work on $\mc E_{\mathrm{agr}}$ and fix $d\ge1$ distinct strings $x_1,\dots,x_d$.  Let
\[
  B\;\defeq\;\bigl\{\,b\in\textstyle\bigcup_{\lambda\ge1}\{0,1\}^{\lambda}:
  x_1,\dots,x_d\in\Lb\,\bigr\}
\]
be the consistent-seed set.  If $B\ne\varnothing$, the examples share a unique length tag $\lambda_F$ and $B\subseteq\{0,1\}^{\lambda_F}$, because a language contains only strings bearing its own length tag.  Three cases.
\begin{itemize}
\item $|B|=0$: the version space is empty and the closure is $\bot$.
\item $|B|=1$, say $B=\{b\}$: the closure is $\Lb$, which is infinite.
\item $|B|\ge 2$: a string $\tuple{\lambda_F,i,y}$ lies in every $\Lb$, $b\in B$, iff $H(b,i)=y$ for all $b\in B$.  Define $J_B\defeq\{i\in\N:H(b,i)\text{ is equal for all }b\in B\}$.  Then, fixing any $b_0\in B$,
  \[
    \cl(\{x_1,\dots,x_d\})\;=\;\{\,E(b_0,i):i\in J_B\,\}.
  \]
Although this realized $B$ may depend on $H$, it is one of the deterministic seed sets covered simultaneously by $\mc E_{\mathrm{agr}}$.  Thus $J_B$ is infinite and so is the closure.
\end{itemize}
Hence, on $\mc E_{\mathrm{agr}}$, no nonempty example set has a nonempty finite closure, and $\Cdim(\Cstar)=0$.

\smallskip
\textbf{(S3) Unpredictability.} We prove the bound $T\,2^{-\lambda}+2^{-\lambda}$ (\cref{lem:s3-formal}), which implies the stated $(T+2)\,2^{-\lambda}$.  The one idea is that adaptivity buys nothing: coupling $A$'s run against an oracle with the samples \emph{planted} at the hidden seed makes the set of seeds $A$ ever tests manifestly independent of $b$, so no adaptive strategy finds $b$ faster than blind guessing.
\end{proof}

\begin{lemma}[Fresh-point unpredictability under samples]
\label[lemma]{lem:s3-formal}
Fix $\lambda$, and let $H$ be the full tagged random oracle of \cref{def:object}, so that $H(\beta,i)\in\{0,1\}^{|\beta|}$ for every seed $\beta$.  Let $b\in\{0,1\}^{\lambda}$ be a uniform seed independent of $H$, and let $Q\subseteq\N$ be a \emph{finite, fixed} index set (chosen independently of $(b,H)$), with samples $S=\{(i,H(b,i)):i\in Q\}$. Every possibly randomized algorithm $A$, with private coins $\rho_A$, making at most $T$ adaptive queries to $H$ satisfies
\[
  \P_{b,H,\rho_A}\!\left[\,A^{H}(S)=\tuple{\lambda,j,w}\ \text{with}\ j\notin Q\ \text{and}\ w=H(b,j)\,\right]
  \;\le\; T\,2^{-\lambda}+2^{-\lambda}.
\]
In particular this is $\negl(\lambda)$ for $T=\poly(\lambda)$ and $\le 2^{-\Omega(\lambda)}$ for every $T\le 2^{(1-\delta)\lambda}$, $\delta>0$; in particular it is at most $(T+2)2^{-\lambda}$, the convenient bound used throughout.
\end{lemma}

\Cref{fig:planting-coupling} summarizes the coupling that makes the bound robust to adaptive queries.

\begin{figure}[H]
\centering
\begin{tikzpicture}[
  font=\footnotesize,
  box/.style={draw=black!55, rounded corners=2pt, align=center, inner sep=5pt,
    minimum height=1.05cm},
  sourcebox/.style={box, fill=black!4, text width=4.25cm},
  oraclebox/.style={box, fill=LinkColor!7, text width=4.35cm},
  runbox/.style={box, fill=black!3, text width=4.25cm},
  seedbox/.style={box, fill=LinkColor!5, text width=2.75cm},
  outbox/.style={box, fill=CiteColor!6, text width=3.7cm},
  flow/.style={-{Latex[length=2mm]}, line width=.65pt, draw=black!70}
]
\node[sourcebox] (sources) {\textbf{Independent sources}\\
  $Q$ fixed; $b,g,(y_i)_{i\in Q},\rho_A$ independent\\[-1pt]
  $S=\{(i,y_i):i\in Q\}$};
\node[oraclebox, right=9mm of sources] (oracle) {\textbf{Plant at the hidden row}\\
  $H(b,i)=y_i$ for $i\in Q$\\[-1pt]
  $H(\beta,i)=g(\beta,i)$ otherwise};
\node[runbox, below=13mm of sources] (shadow) {\textbf{Shadow execution}\\
  run $A^g(S)$ with coins $\rho_A$\\[-1pt]
  transcript independent of $b$};
\node[seedbox, right=9mm of shadow] (seeds) {\textbf{Tested rows}\\
  $\mathrm{Seeds}\perp b$\\[-1pt]
  $|\mathrm{Seeds}|\le T$};
\node[outbox, right=9mm of seeds, yshift=6.5mm] (hit) {\textbf{True row tested}\\
  $\P[b\in\mathrm{Seeds}]\le T2^{-\lambda}$};
\node[outbox, right=9mm of seeds, yshift=-6.5mm] (guess) {\textbf{No true-row query}\\
  fresh-value guess $\le 2^{-\lambda}$};

\draw[flow] (sources) -- (oracle);
\draw[flow] (sources) -- (shadow);
\draw[flow] (shadow) -- (seeds);
\draw[flow] (seeds) -- (hit);
\draw[flow] (seeds) -- (guess);
\draw[flow, densely dashed] (oracle.south) -- node[left, font=\scriptsize, align=right]
  {coupled to the shadow run\\before the first seed-$b$ query} (seeds.north);
\end{tikzpicture}
\caption{The planted coupling behind \cref{lem:s3-formal}.  Planting preserves the law of a random
oracle independent of the hidden seed, while the shadow transcript and its tested-row set remain independent of that seed.  Success is bounded by either testing row $b$ or guessing a fresh $\lambda$-bit oracle value.}
\label[figure]{fig:planting-coupling}
\end{figure}
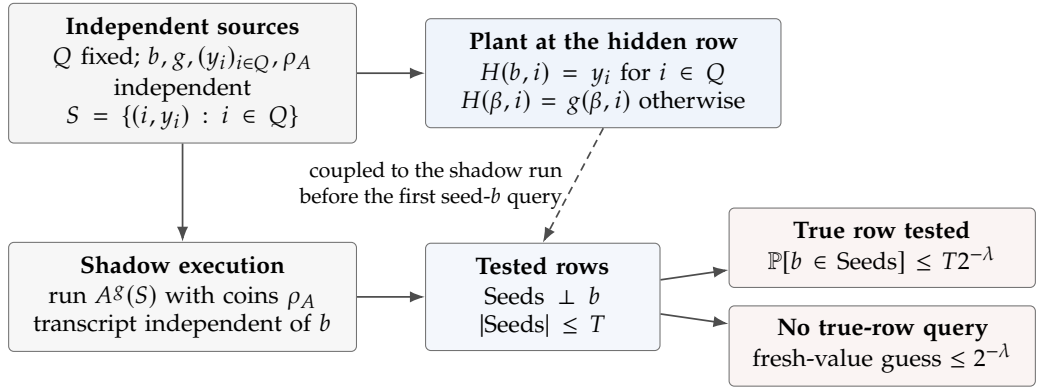

\begin{proof}
Sample four independent sources: sample values $(y_i)_{i\in Q}$ i.i.d.\ uniform on $\{0,1\}^{\lambda}$; a uniformly random \emph{generic oracle} $g$ on the full tagged domain, with $g(\beta,i)\in\{0,1\}^{|\beta|}$; $b$ uniform on $\{0,1\}^{\lambda}$; and $A$'s private random tape $\rho_A$ (degenerate if $A$ is deterministic), all mutually independent. Define the actual oracle by \emph{planting} the samples at the (unknown) seed $b$:
\[
  H(\beta,i)\;\defeq\;
  \begin{cases}
    y_i & \text{if } \beta=b \text{ and } i\in Q,\\
    g(\beta,i) & \text{otherwise.}
  \end{cases}
\]

\emph{Faithfulness.} Fix any $\beta_0$; then $H$ is $g$ with its $|Q|$ values on $\{(\beta_0,i):i\in Q\}$ overwritten by independent uniform $y_i$, which is again a uniform random function whose law is independent of $\beta_0$. Hence $H$ is a random oracle independent of $b$, $H(b,i)=y_i$ for $i\in Q$, and $H$ agrees with $g$ off the planted set $P\defeq\{(b,i):i\in Q\}$; every planted input carries seed $b$.

\emph{Shadow run.} Run $A$ with coins $\rho_A$ and input $S$, answering queries with $g$ instead of $H$. Since $(g,(y_i),\rho_A)\perp b$, the shadow transcript is a function of $(g,(y_i),\rho_A)$, hence independent of $b$. Let $\mathrm{Seeds}\defeq\{\beta\in\{0,1\}^{\lambda}: A\ \text{queries some}\ (\beta,i)\ \text{in the shadow run}\}$; then $|\mathrm{Seeds}|\le T$ and $\mathrm{Seeds}\perp b$.  Queries with other seed lengths are unplanted and affect neither the coupling nor this count.

\emph{Coupling.} Run $A$ with the true $H$ and the same $\rho_A$. The runs receive identical answers at every input off $P$, and every input of $P$ has seed $b$. Let $k$ be the first step at which the real run queries a seed-$b$ input ($k=\infty$ if none). Through step $k-1$ the real run queries only non-seed-$b$ inputs, where $H=g$; as $A$ is deterministic given its view, the real and shadow runs issue identical queries and receive identical answers through step $k-1$, so their step-$k$ queries coincide.

\emph{(i) The true-seed query.} Let $\mathsf{QT}$ be the event that the real run ever queries seed $b$. On $\mathsf{QT}$ the step-$k$ query has seed $b$ and equals the shadow run's, so $b\in\mathrm{Seeds}$; thus $\mathsf{QT}\subseteq\{b\in\mathrm{Seeds}\}$ and, as $\mathrm{Seeds}\perp b$ has size $\le T$,
\[
  \P[\mathsf{QT}]\le\P[b\in\mathrm{Seeds}]=\E[|\mathrm{Seeds}|]\,2^{-\lambda}\le T\,2^{-\lambda}.
\]

\emph{(ii) Without it.} Condition on $b=\beta_0$. On $\lnot\mathsf{QT}$ the real run never queries seed $\beta_0$, so it equals the shadow run and $A$'s output $\tuple{\lambda,j,w}$ is a function of the shadow transcript. Success needs $j\notin Q$, whence $H(\beta_0,j)=g(\beta_0,j)$ is unplanted and queried by neither run, hence uniform on $\{0,1\}^{\lambda}$ conditionally on $A$'s view; as $(j,w)$ is a function of that view, $\P[w=H(\beta_0,j)\mid\text{view}]=2^{-\lambda}$. So $\P[\text{success}\wedge\lnot\mathsf{QT}]\le 2^{-\lambda}$.

Combining, $\P[\text{success}]\le\P[\mathsf{QT}]+\P[\text{success}\wedge\lnot\mathsf{QT}]\le T2^{-\lambda}+2^{-\lambda}$.
\end{proof}

\begin{remark}[Adaptivity and decoys]\label[remark]{rem:adaptive-decoy}
Adaptive consistency tests only enlarge the $b$-independent set $\mathrm{Seeds}$; they do not increase the chance that it contains $b$.  A consistent decoy $\beta'\ne b$ is equally unhelpful because $H(\beta',j)$ is independent of the fresh target value $H(b,j)$.  The coupling is essential: independence of $\mathrm{Seeds}$ from $b$, not an exchangeability heuristic, handles adaptivity.
\end{remark}

\begin{lemma}[Zero-mistake information-theoretic generation]
\label[lemma]{lem:easy}
Almost surely over $H$, a computationally unbounded generator makes zero mistakes on every target $\Lb$, under every complete distinct enumeration.
\end{lemma}

The following oracle procedure makes the unbounded construction explicit.  On the probability-one event established in the proof below, it halts on every valid generation-game history; the open-ended index scan is intentional.

\begin{algorithm}[H]
\caption{Unbounded zero-mistake generation}
\label[algorithm]{alg:zero-mistake}
\begin{algorithmic}[1]
\Require $1^\lambda$, valid history $x_s=\tuple{\lambda,i_s,y_s}$ for $1\le s\le t$, and oracle access to $H$
\State $F_t\gets\{x_1,\ldots,x_t\}$ and $I_t\gets\{i_1,\ldots,i_t\}$
\State $B_t\gets\{\beta\in\{0,1\}^{\lambda}:F_t\subseteq L_\beta\}$
  \Comment{exhaustive version space; $B_0=\{0,1\}^{\lambda}$}
\For{$j\in\N\setminus I_t$ in increasing order} \Comment{unbounded search}
  \State Query $H(\beta,j)$ for every $\beta\in B_t$
  \If{$\{H(\beta,j):\beta\in B_t\}=\{y\}$ for some $y$}
    \State \Return $\tuple{\lambda,j,y}$
  \EndIf
\EndFor
\end{algorithmic}
\end{algorithm}

\begin{proof}
Work on the event $\mc E_{\mathrm{agr}}$ constructed in the proof of (S2).  At the cold start, $B_0=\{0,1\}^{\lambda}$ has infinitely many common indices, so \cref{alg:zero-mistake} finds an output valid for every length-$\lambda$ target.  After examples arrive, the consistent seed set $B_t$ is finite; if it is a singleton every fresh index works, and otherwise (S2) gives infinitely many common indices on $\mc E_{\mathrm{agr}}$.  Excluding the finite revealed set therefore leaves a valid output at every round.  The scans may be arbitrarily expensive, as the information-theoretic model permits.
\end{proof}

\begin{remark}[Why the empty-history output matters]\label[remark]{rem:cold-start}
The equality $\Cdim(\Cstar)=0$ constrains closures of nonempty example sets; by itself it says nothing about the intersection of all targets in the length-$\lambda$ game before any example is revealed. Thus the forced output $a_0$ in \cref{def:model} could still contribute one mistake even if every later output were closure-safe.  The separate infinite common-intersection property proved above makes $a_0$ valid uniformly over the target seed and removes exactly this possible initial mistake.  This distinction sharpens the information-theoretic bound from one to zero; it is not used by the computational lower-bound mechanism.
\end{remark}

\begin{lemma}[Computational hardness]
\label[lemma]{lem:hard}
Let $\sigma_b$ denote the canonical ascending enumeration $E(b,1),E(b,2),\dots$ of $\Lb$. For every uniform polynomial-query generator $G$ and every nonnegative polynomial $p$, put $m_p=\lceil p(\lambda)\rceil$ and $W_p=2(m_p+1)$.  The \emph{finite-horizon joint seed--oracle bound}
\[
  \P_{b,H}\bigl[\,\E_G[\mist_{<W_p}(G,\Lb,\sigma_b)]\le m_p\,\bigr]
  \;\le\;\negl(\lambda)
\]
holds, where the outer probability is over a uniform seed $b\in\{0,1\}^{\lambda}$ and the oracle $H$; the expectation has already averaged over $G$'s coins. Consequently, for every polynomial $p$ and all sufficiently large $\lambda$, all but a negligible fraction of seed--oracle pairs $(b,H)$ have more than $m_p\ge p(\lambda)$ expected mistakes already in their first $W_p$ outputs.  In particular, the same lower bound holds for total mistakes.  The exponential version is quantified in \cref{prop:perseed-exp}.
\end{lemma}

The reduction turns a mistake-bounded generator into a fresh-point predictor: by (S1), a non-mistake on the graph is a correct target value at an index unrevealed by the sample history, which (S3) makes hard to produce.

\begin{proof}
Fix a uniform $\poly$-query generator $G$ and a polynomial $p$; let $m\defeq m_p=\lceil p(\lambda)\rceil$ and $W\defeq W_p=2(m+1)$. Under the canonical enumeration $\sigma_b$ the revealed set at step $t$ is $\{E(b,i):1\le i\le t\}$ (empty at $t=0$), and $G$ produces $a_0=G(1^\lambda,\varnothing)$ and $a_t=G(1^\lambda,x_1,\dots,x_t)$ using its own oracle queries.

\smallskip
\textbf{Step 1 (a non-mistake is a fresh-index prediction).} If $a_t$ is not a mistake, then $a_t\in \Lb\setminus\{x_1,\dots,x_t\}$. Membership in $\Lb$ forces $a_t=\tuple{\lambda,j,H(b,j)}=E(b,j)$ for its index field $j$, by (S1); and $a_t\notin\{E(b,1),\dots,E(b,t)\}$ forces $j>t$. So a non-mistake at step $t$ is exactly a correct output $E(b,j)$ at a fresh index $j>t$.

\smallskip
\textbf{Step 2 (the mistake budget forces a dense band).} Fix a seed--oracle pair $(b,H)$ with
\[
  \E_G[\mist_{<W}(G,\Lb,\sigma_b)]\le m.
\]
Since the window contains $W$ outputs,
\[
  \E_G[\#\{\text{non-mistakes in }0,\dots,W-1\}]\ge W-m=m+2.
\]
Drawing $t^{*}$ uniformly from $\{0,\dots,W-1\}$,
\begin{equation}\label{eq:band}
  \P_{G,t^{*}}[\,a_{t^{*}}\text{ is a non-mistake}\,]
  \;=\;\frac1W\sum_{t=0}^{W-1}\P_G[\text{non-mistake at }t]
  \;\ge\;\frac{m+2}{2(m+1)}\;\ge\;\frac12.
\end{equation}

\smallskip
\textbf{Step 3 (the predictor).} Define the predictor $A_\lambda$ for the game of~(S3): draw $t^{*}$ uniform in $\{0,\dots,W-1\}$; set $Q=\{1,\dots,t^{*}\}$ (empty if $t^{*}=0$) and receive $\{E(b,i):i\in Q\}$; draw $G$'s coins internally and feed $G$ the canonical prefix histories $\varnothing,(x_1),\dots,(x_1,\dots,x_{t^{*}})$, collecting $a_0,\dots,a_{t^{*}}$; output $\hat y\defeq a_{t^{*}}$. All of $A_\lambda$'s oracle queries are $G$'s own simulated queries on this prefix, so $A_\lambda$ is a legitimate predictor with the seed hidden and $|Q|\le W-1$. By Step 1, whenever $a_{t^{*}}$ is a non-mistake we have $\hat y=E(b,j)$ with $j>t^{*}$, hence $j\notin Q$: exactly the success event of (S3).

\smallskip
Let
\[
  B_\lambda\defeq\{(b,H):\E_G[\mist_{<W}(G,\Lb,\sigma_b)]\le m\}.
\]
Combining Step 2 (on each $(b,H)\in B_\lambda$, the conditional success is $\ge\tfrac12$) with Step 1,
\[
  \P_{b,H,\,\mathrm{coins},\,t^{*}}[\,A_\lambda\text{ outputs }E(b,j),\ j\notin Q\,]
  \;\ge\;\tfrac12\,\P_{b,H}[\,(b,H)\in B_\lambda\,].
\]
The cut $t^{*}$ is drawn from $A_\lambda$'s own coins, independently of $(b,H)$; so we may condition on $t^{*}$, apply \cref{lem:s3-formal} to the fixed index set $Q=\{1,\dots,t^{*}\}$, and average. Its bound is at most $(T+2)2^{-\lambda}$, uniform in $Q$, hence survives the averaging and caps $A_\lambda$'s success. By Step 4 below $A_\lambda$ makes $T=\poly(\lambda)$ queries, so the cap is negligible. Therefore
\[
  \P_{b,H}[(b,H)\in B_\lambda]\le 2(T+2)2^{-\lambda}=\negl(\lambda),
\]
which is exactly the joint bound of \cref{lem:hard}.

\smallskip
\textbf{Step 4 ($A_\lambda$ uses polynomially many queries).} $A_\lambda$ simulates $G$ through $t^{*}+1\le W=2(m+1)=\poly(\lambda)$ generation rounds.  By the per-output query cap,
\[
  T\le\sum_{i=0}^{W-1}q_G(\lambda,i)
  \le W\cdot\poly(\lambda,W)=\poly(\lambda).
\]
Each simulated output computation halts by \cref{def:model}; no bound on its internal running time is used or claimed.

\smallskip
The displayed bound proves the stated consequence; replacing the polynomial threshold by an exponential one gives the next proposition without any infinite-horizon inference.
\end{proof}

The reduction is not tied to a polynomial window.  With an exponential window, it makes the floor explicit and quantifies how rare an easy seed is.

\begin{proposition}[Joint seed--oracle exponential floor]\label[proposition]{prop:perseed-exp}
Let $G$ be a uniform generator making at most $q_G(\lambda,i)$ oracle queries while producing its $i$th output.  Fix an integer $d\ge0$ and a polynomial $r_G$ such that $q_G(\lambda,i)\le r_G(\lambda)(i+1)^d$, and let $0<c<1/(d+1)$.  Then there are $\gamma=\gamma(G,c)>0$ and $\lambda_0(G,c)$ with, for all $\lambda\ge\lambda_0$,
\[
  \P_{b,H}\!\left[
    \E_G\!\left[\mist_{<W_c}(G,\Lb,\sigma_b)\right]\le m_c
  \right]\;\le\;2^{-\gamma\lambda},
  \qquad m_c\defeq\lceil2^{c\lambda}\rceil,\quad W_c\defeq 2(m_c+1).
\]
\end{proposition}
\begin{proof}
Instantiate the reduction of \cref{lem:hard} with $m(\lambda)=m_c$ in place of a polynomial; the derivation is window-agnostic and uses only $W=2(m+1)$ and the threshold $m$. The predictor makes $T\le W\,r_G(\lambda)(W+1)^d=2^{c(d+1)\lambda+o(\lambda)}$ queries, so $0<c<1/(d+1)$ gives $T<2^{\lambda-1}$ for large $\lambda$; by \cref{lem:s3-formal} the predictor wins with probability $\le(T+2)2^{-\lambda}\le2^{-[1-c(d+1)]\lambda+o(\lambda)}$. Including the factor two from the dense-band reduction, the easy seed--oracle event has probability at most $2^{-\gamma\lambda}$ for any fixed $\gamma<1-c(d+1)$.
\end{proof}

We next convert the joint seed--oracle bound into an almost-sure statement for a single realized oracle.

\begin{theorem}[Almost-sure exponential hardness]\label[theorem]{thm:as-exp}
There is an event $\mc{E}^{\ast}$ with $\P_H[\mc{E}^{\ast}]=1$ on which, for every uniform polynomial-query generator $G$, there are $c_G>0$ and $\Lambda(G,H)<\infty$ such that for every $\lambda\ge\Lambda(G,H)$ some seed $b\in\{0,1\}^{\lambda}$ has, for $W_G(\lambda)=2(\lceil2^{c_G\lambda}\rceil+1)$, $\E_G[\mist_{<W_G(\lambda)}(G,\Lb,\sigma_b)]>2^{c_G\lambda}$.  Thus the exponential worst-case mistake floor is incurred within an explicit exponential horizon.
\end{theorem}

\begin{proof}
\textbf{Setup.} Fix a uniform polynomial-query generator $G$.  Choose an integer $d_G\ge0$ and a polynomial $r_G$ witnessing
\[
  q_G(\lambda,i)\le r_G(\lambda)(i+1)^{d_G},
\]
and set
\[
  c_G\defeq\frac{1}{2(d_G+1)},\qquad
  m_G(\lambda)\defeq\left\lceil2^{c_G\lambda}\right\rceil,\qquad
  W_G(\lambda)\defeq2\bigl(m_G(\lambda)+1\bigr).
\]
For a length-$\lambda$ seed--oracle pair, define
\[
  Z_{\lambda,G}(b,H)
  \defeq
  \E_G\!\left[\mist_{<W_G(\lambda)}(G,\Lb,\sigma_b)\right],
  \qquad
  \mathsf{Easy}_{\lambda,G}
  \defeq
  \bigl\{(b,H):Z_{\lambda,G}(b,H)\le m_G(\lambda)\bigr\}.
\]
For a fixed oracle, let
\[
  f_{\lambda,G}(H)
  \defeq
  \P_b\!\left[(b,H)\in\mathsf{Easy}_{\lambda,G}\right].
\]
Thus $f_{\lambda,G}(H)$ is exactly the fraction of length-$\lambda$ seeds that are easy for the realized oracle $H$.  We now convert the joint seed--oracle estimate into a fixed-oracle worst-case statement in five explicit stages.

Let $\gamma_G>0$ and $\lambda_0(G)$ be supplied by \cref{prop:perseed-exp} with $c=c_G$.

\smallskip
\textbf{Step 1 (joint easy-pair bound).}\enspace Because $c_G<1/(d_G+1)$, for every $\lambda\ge\lambda_0(G)$,
\[
  \P_{b,H}\!\left[\mathsf{Easy}_{\lambda,G}\right]
  \le 2^{-\gamma_G\lambda}.
\]

\smallskip
\textbf{Step 2 (average easy-seed fraction).}\enspace Since $f_{\lambda,G}(H)$ is the conditional probability of $\mathsf{Easy}_{\lambda,G}$ after fixing $H$, Tonelli's theorem (\cref{fact:tonelli}) gives
\[
  \E_H\!\left[f_{\lambda,G}(H)\right]
  =\P_{b,H}\!\left[\mathsf{Easy}_{\lambda,G}\right]
  \le2^{-\gamma_G\lambda}.
\]
The random seed is used only to measure this fraction; after this step, the exceptional event is an event over the oracle alone.

\smallskip
\textbf{Step 3 (few oracles have many easy seeds).}\enspace Markov's inequality (\cref{fact:markov}) at threshold $\lambda^{-2}$ yields
\[
  \P_H\!\left[f_{\lambda,G}(H)\ge\lambda^{-2}\right]
  \le \lambda^2\E_H\!\left[f_{\lambda,G}(H)\right]
  \le \lambda^2 2^{-\gamma_G\lambda}.
\]

\smallskip
\textbf{Step 4 (fix almost every oracle).}\enspace The bound is summable:
\[
  \sum_{\lambda\ge\lambda_0(G)}\lambda^2 2^{-\gamma_G\lambda}<\infty.
\]
Hence the first Borel--Cantelli lemma (\cref{fact:first-borel-cantelli})---which requires no independence across lengths---gives a probability-one event
\[
  \mc E_G
  \defeq
  \left\{H:
  f_{\lambda,G}(H)<\lambda^{-2}
  \text{ for all sufficiently large }\lambda\right\}.
\]

\smallskip
\textbf{Step 5 (extract a hard seed and make the conclusion simultaneous).}\enspace Fix $H\in\mc E_G$ and choose $\Lambda(G,H)$ so that, for every $\lambda\ge\Lambda(G,H)$, one has $\lambda\ge\max\{\lambda_0(G),2\}$ and $f_{\lambda,G}(H)<\lambda^{-2}<1$.  More than a $1-\lambda^{-2}$ fraction of length-$\lambda$ seeds then lie outside $\mathsf{Easy}_{\lambda,G}$; in particular, some $b=b(\lambda,G,H)$ satisfies
\[
  Z_{\lambda,G}(b,H)>m_G(\lambda)\ge2^{c_G\lambda}.
\]
By the definition of $Z_{\lambda,G}$, this is the required expected mistake floor within $W_G(\lambda)=2(\lceil2^{c_G\lambda}\rceil+1)$ outputs.

Finally, uniform polynomial-query oracle machines have finite descriptions and therefore form a countable family.  For each such $G$, fix the choices above before fixing $H$ and construct the corresponding full-measure event $\mc E_G$.  The countable full-measure intersection (\cref{fact:countable-intersection}) gives
\[
  \P_H[\mc E^\ast]=1,
  \qquad
  \mc E^\ast\defeq\bigcap_G\mc E_G.
\]
On this single event, the preceding conclusion holds for every $G$ with its own $c_G$ and $\Lambda(G,H)$.
\end{proof}

Combining this almost-sure hardness event with information-theoretic easiness gives the main theorem.

\begin{theorem}[The computational price of generation]
\label[theorem]{thm:main}
Relative to a random oracle $H$, the collection $\Cstar$ satisfies, almost surely over $H$:
\begin{enumerate}
\item[\textup{(i)}] \textbf{\textup{(Recognizable.)}}\enspace Given the seed $b$, membership ``$x\in\Lb$?'' is decidable in $\poly(|b|+|x|)$ time with a single query to $H$; each $\Lb$ is infinite and, given $b$, polynomial-time enumerable with oracle access to $H$.
\item[\textup{(ii)}] \textbf{\textup{(Information-theoretically easy.)}}\enspace $\Cdim(\Cstar)=0$; using the common lengthwise intersection, a computationally unbounded generator makes \emph{zero} mistakes on every target $\Lb$, under every complete distinct enumeration.
\item[\textup{(iii)}] \textbf{\textup{(Computationally hard.)}}\enspace For every uniform polynomial-query generator $G$ there is $c_G>0$ such that, at all but finitely many lengths $\lambda$, some seed $b\in\{0,1\}^{\lambda}$ satisfies
  \[
    \E_G\!\left[\mist_{<W_G(\lambda)}(G,\Lb,\sigma_b)\right]>2^{c_G\lambda},
    \qquad W_G(\lambda)=2(\lceil2^{c_G\lambda}\rceil+1),
  \]
where $\sigma_b$ is the canonical complete enumeration.
\end{enumerate}
In particular $\Cstar$ separates zero information-theoretic mistakes from a $2^{\Omega(\lambda)}$ worst-case expected finite-horizon mistake floor for every fixed uniform polynomial-query (a fortiori polynomial-time) generator.  The exponent is generator-dependent, and the statement is relative to the random oracle.
\end{theorem}

\begin{proof}
Part~(i) is (S1) of \cref{lem:construction} and holds for every oracle.  Let $\mc{E}_{\mathrm{easy}}$ be the event on which $\Cdim(\Cstar)=0$ and, for every $\lambda$, the intersection of all length-$\lambda$ graph languages is infinite.  Property~(S2) and the first paragraph of \cref{lem:easy}'s proof show $\P_H[\mc{E}_{\mathrm{easy}}]=1$.  On this event, \cref{lem:easy} gives the unbounded generator zero mistakes, proving part~(ii).  The event $\mc E^*$ of \cref{thm:as-exp} has probability one and gives part~(iii) simultaneously for every uniform polynomial-query generator.  Their intersection $\mc{E}_{\mathrm{easy}}\cap\mc E^*$ still has probability one, so all three parts hold for the same realized oracle.
\end{proof}

\begin{remark}[The exponent is per-generator]\label[remark]{rem:perG-const}
Under the per-output budget, the admissible $c_G<1/(d+1)$ depends on $G$ through the degree $d$ of its round-index dependence.  As $d$ is unbounded across generators, no single constant serves every $G$: ``$2^{\Omega(\lambda)}$'' is inherently per-$G$, exactly as \cref{thm:main}(iii) states.  If per-output queries are $\poly(\lambda)$ independently of the round index, then $d=0$ and any $c<1$ is admissible. \Cref{prop:no-universal-exponent} shows that generator dependence is necessary, not merely an artifact of this proof.
\end{remark}

\begin{remark}[Quantifier order]\label[remark]{rem:notclaimed}
Here and below, $\forall^{\infty}\lambda$ means ``for all sufficiently large $\lambda$.'' The lower bound in \cref{thm:main}(iii) has the order
{\small
\[
  \P_H\!\left[
  \begin{gathered}
    \forall\text{ uniform polynomial-query }G\;
    \exists c_G>0\;
    \forall^{\infty}\lambda\\[-1pt]
    \exists b\in\{0,1\}^{\lambda}:\quad
    \E_G\!\left[\mist_{<W_G(\lambda)}(G,\Lb,\sigma_b)\right]
      >2^{c_G\lambda}
  \end{gathered}
  \right]=1,
\]
}
where $W_G(\lambda)=2(\lceil2^{c_G\lambda}\rceil+1)$.  Thus the probability-one oracle event is shared by all uniform machines, while the exponent and the hard target may depend on $G$ (and the target also on $\lambda$).  In particular, the order $\forall G\,\forall^{\infty}\lambda$ is distinct from $\exists^{\infty}\lambda\,\forall G$, which would require one infinite length set to work for all generators at once; the two orders are incomparable.
\end{remark}

\pagebreak[3]
\begin{remark}[Average case and target information]\label[remark]{rem:average}
The joint bound in \cref{prop:perseed-exp} makes the easy seed--oracle pairs exponentially rare.  After fixing an almost-sure oracle by Markov and Borel--Cantelli, the proof yields a $1-o(1)$ hard-seed fraction (specifically $1-\lambda^{-2}$) for every fixed uniform generator at all large lengths.  A target-aware generator given $b$ can enumerate $\Lb$ with no post-start mistakes.  This observation does not settle hardness against target-independent per-length advice or circuit families.
\end{remark}

\section{Beyond the canonical enumeration}
\label[section]{sec:beyond-canonical}

The canonical ascending enumeration used in \cref{lem:hard}, and hence in \cref{thm:main}, is convenient, but the order itself is not load-bearing.  By property~\textup{(S1)} of \cref{lem:construction}, every complete distinct enumeration of $\Lb$ is induced by a permutation of its index set.  What the hardness reduction needs is not the identity permutation, but an index order fixed independently of the hidden target and oracle.

\begin{definition}[Precommitted index schedule]
\label[definition]{def:precommitted-schedule}
At length $\lambda$, a \emph{precommitted index schedule} is a possibly random permutation
\[
  \Pi_\lambda=(\Pi_{\lambda,1},\Pi_{\lambda,2},\ldots)
\]
of $\N$ that is independent of the hidden seed $b$, the random oracle $H$, and the private coins of $G$.  It induces the complete distinct enumeration
\[
  \sigma_{b,\Pi}
  =
  \bigl(E(b,\Pi_{\lambda,1}),
        E(b,\Pi_{\lambda,2}),\ldots\bigr).
\]
A schedule family $\Pi=(\Pi_\lambda)_\lambda$ is \emph{uniform} if it is produced by one oracle-free fair-coin probabilistic Turing machine which, on input $1^\lambda$ and a private random tape, outputs every finite prefix of $\Pi_\lambda$ in finite time.  The scheduler's tape is independent of $(b,H)$ and of $G$'s coins.  Deterministic permutations are point-mass special cases.
\end{definition}

Only the index order is independent of $(b,H)$: the enumerated strings $E(b,\Pi_{\lambda,t})$ necessarily depend on both.

\begin{restatable}[Schedule-robust exponential floor]{theorem}{ScheduleRobustness}
\label[theorem]{thm:schedule-robust}
Let $G$ be a uniform generator satisfying
\[
  q_G(\lambda,i)\le r_G(\lambda)(i+1)^d
\]
for a polynomial $r_G$ and an integer $d\ge0$.  Fix $0<c<1/(d+1)$ and put
\[
  m_c=\lceil2^{c\lambda}\rceil,
  \qquad
  W_c=2(m_c+1).
\]
There are $\gamma=\gamma(G,c)>0$ and $\lambda_0(G,c)$, independent of the schedule, such that every precommitted schedule family $\Pi$ satisfies, for all $\lambda\ge\lambda_0$,
\[
  \P_{b,H}\!\left[
    \E_{G,\Pi}\!\left[
      \mist_{<W_c}(G,\Lb,\sigma_{b,\Pi})
    \right]\le m_c
  \right]
  \le 2^{-\gamma\lambda}.
\]
Consequently, for each fixed precommitted schedule family $\Pi$, almost surely over $H$, every sufficiently large $\lambda$ has some $b\in\{0,1\}^{\lambda}$ such that
\[
  \E_{G,\Pi}\!\left[
    \mist_{<W_c}(G,\Lb,\sigma_{b,\Pi})
  \right]
  >m_c\ge2^{c\lambda}.
\]

Moreover, let $\mathfrak S_{\mathrm{unif}}$ be the countable class of uniform precommitted schedulers.  Then
\[
  \P_H\!\left[
    \begin{gathered}
      \forall\text{ uniform polynomial-query }G\;
      \exists c_G>0\;
      \forall\Pi\in\mathfrak S_{\mathrm{unif}}\;
      \forall^{\infty}\lambda\;
      \exists b\in\{0,1\}^{\lambda}:\\[-2pt]
      \E_{G,\Pi}\!\left[
        \mist_{<W_G(\lambda)}(G,\Lb,\sigma_{b,\Pi})
      \right]>2^{c_G\lambda}
    \end{gathered}
  \right]=1,
\]
where $W_G(\lambda)=2(\lceil2^{c_G\lambda}\rceil+1)$.  The exponent $c_G$ may be selected from a query bound for $G$ alone; the cutoff implicit in $\forall^{\infty}\lambda$ may depend on $(G,\Pi,H)$.
\end{restatable}

Combining \cref{thm:schedule-robust} with \cref{thm:main}(ii) places the two regimes on one probability-one oracle event: unbounded generation still makes zero mistakes under every complete distinct enumeration, while every uniform polynomial-query generator has an exponential worst-case mistake floor simultaneously for all uniform precommitted schedulers.

\begin{proof}[Proof idea]
For a uniformly random $t^*\in\{0,\ldots,W-1\}$, the stopped output satisfies
\[
  \P_{G,\Pi,t^*}[a_{t^*}\text{ is a non-mistake}]
  =
  1-\frac1W
  \E_{G,\Pi}\!\left[
    \mist_{<W}(G,\Lb,\sigma_{b,\Pi})
  \right]
\]
for every fixed $(b,H)$.  Conditioned on $t^*$ and the realized schedule prefix, precommitment makes its index set fixed and independent of $(b,H)$. The stopped simulation is therefore admissible for \cref{lem:s3-formal} and uses at most $T_W=\sum_{i=0}^{W-1}q_G(\lambda,i)$ queries.  Consequently,
\[
  \P_{b,H}\!\left[
    \E_{G,\Pi}[\mist_{<W}]\le m
  \right]
  \le
  \frac{T_W+1}{1-m/W}\,2^{-\lambda}.
\]
For $m=m_c$ and $W=W_c$, one has $1-m/W>1/2$ and $T_W=2^{c(d+1)\lambda+o(\lambda)}$.  The remaining Tonelli--Markov--Borel--Cantelli conversion is identical in structure to \cref{thm:as-exp}.  The complete proof, including conditioning on the ordered prefix, appears in \cref{app:schedule-proofs}.
\end{proof}

\begin{remark}[Scope of schedule robustness]
\label[remark]{rem:schedule-scope}
Precommitment is essential.  The theorem does not cover target-dependent, oracle-dependent, or sample-adaptive index orders: a target-dependent first index can encode $b$, after which one oracle query per round suffices to generate fresh target points.  Simultaneous coverage is therefore only over the countable class of uniform precommitted schedulers, not all abstract permutations.  For randomized schedulers, the expectation includes both scheduler and generator coins.
\end{remark}

\section{Discussion}
\label[section]{sec:discussion}

The separation exposes a distinction invisible to closure dimension alone: safe outputs may exist abundantly yet remain computationally hard to locate.  The random-function graph makes the gap exact---unbounded search exploits recurring common points, while sparse queries leave a fresh target value hidden.

\paragraph{What the lower bound says.} The theorem does not assert eventual failure.  Generation in the limit imposes no computational bound~\citep{kleinberg2024generation}; our finite horizon instead measures the mistakes accumulated before an efficient generator can locate fresh target points.  This differs from rank-dependent output deadlines~\citep{ganjuetal2026timesensitive}: the horizon scores cumulative mistakes under a query budget and does not rank target strings.

\paragraph{Scope of the mechanism.} The theorem ranges over uniform machines; giving the target seed $b$ to the generator trivializes the instance, while target-independent advice or circuits remain open.  Its closure argument uses \emph{infinite accidental agreement}, not code distance: every fixed finite set of consistent rows agrees infinitely often almost surely.  Finally, the model is feedback-free.  Unlike online learners or generators with explicit mistake or query feedback~\citep{hannekeetal2026feedback}, the generator never learns whether an output was valid.  Computable online learning exhibits a related combinatorial-versus-computational gap in a different prediction game \citep{hasratibendavid2023computable}.

\paragraph{Broader resource interactions.} Other constraints can change not only whether generation is possible, but also the computational cost of locating the next acceptable output.  Privacy, corrupted or relational evidence, and distinctions among valid outputs lead to three concrete extensions of the present query-complexity question.

\emph{Privacy.}  Privacy constrains how the observed stream may influence released outputs.  In an agnostic statistical formulation, privacy--accuracy tradeoffs are quantified through error-rate exponents \citep{lixiaoyuetal2026privacy}.  Under continual release, every countable class admits $\varepsilon$-private generation, yet uniform private generation for some size-$k$ classes needs $\Omega(k/\varepsilon)$ samples \citep{mehrotraetal2026dp}.  These results suggest a three-way frontier among privacy, cumulative mistakes, and the cost of producing each output.  Does the zero-versus-exponential separation persist under differential privacy, and how should it scale with the privacy parameters?

\emph{Corrupted and relational evidence.}  Corruption changes the evidence rather than the output criterion.  Finite insertions can be benign for some collections yet separate clean from noisy generatability for others \citep{ramanraman2025noisy,baipanigrahizhang2026noise,liaronzhang2026noise}; vanishing-rate infinite contamination gives another robustness regime \citep{mehrotraetal2025contamination}.  Contrastive presentations replace positive examples by unoriented crossing pairs and are characterized by a contrastive closure dimension, with distinct behavior under corruption \citep{lixiaoyuetal2026contrastive}.  Under a query budget, can corruption amplify the cost of locating a safe point even when generation remains information-theoretically possible, or can relational evidence make such points easier to find?

\emph{Value-sensitive generation.}  Our objective treats every fresh target point equally.  In a nested model of valuable mathematics, stronger coverage can require infinitely many valid but trivial outputs, even when their asymptotic rate vanishes \citep{lixiaoyuetal2026floodharvest}.  Verifier-assisted constrained generation has a distinct query-complexity theory in which process-verifier access can turn otherwise intractable sampling tasks into tractable ones \citep{bottaetal2025verifier}.  Once verifier queries and computation are budgeted, what resources are required to locate valuable rather than merely valid outputs?

The most direct extension of the theorem is to replace the random oracle without losing either side of the separation.

\subsection{What would a standard-model separation require?}
\label[subsection]{sec:standard-model-requirements}

A direct standard-model port of the graph-and-stopping argument would follow from one public, deterministic family with five properties.  These conditions are sufficient for this proof strategy, not necessary for every possible separation.

\begin{enumerate}
\item[\textup{(R1)}] \textbf{\textup{(Efficient public evaluation.)}}\enspace A uniform algorithm computes each row value in time polynomial in $\lambda+\log i$; hardness does not come from hiding a language once its seed is known.

\item[\textup{(R2)}] \textbf{\textup{(Infinite version-space common points.)}}\enspace After every nonempty finite realizable history, the consistent seeds share infinitely many graph points, allowing an unbounded generator to avoid the finite revealed set.

\item[\textup{(R3)}] \textbf{\textup{(Cold-start agreement.)}}\enspace Before the first example, all seeds of one length share a graph point.  This extra property, not closure dimension alone, removes the possible initial mistake.

\item[\textup{(R4)}] \textbf{\textup{(Quantitative fresh-point unpredictability.)}}\enspace For every independently chosen finite sample-index set, a seed-oblivious predictor using resource $T$ finds a fresh target point with probability at most $\varepsilon_\lambda(T)$, uniformly over the set and its supplied order.

\item[\textup{(R5)}] \textbf{\textup{(One object against all uniform generators.)}}\enspace One public family satisfies the structural and security properties against every uniform generator; the exponent and cutoff may depend on the generator, but the family may not.
\end{enumerate}

The security scale in \textup{(R4)} is load-bearing.  For $m=\lceil2^{c\lambda}\rceil$, the stopped predictor used by the reduction simulates $W=2(m+1)$ outputs.  If output $i$ costs at most $\poly(\lambda)(i+1)^d$, the simulation may require up to
\[
  T_W \le 2^{c(d+1)\lambda+o(\lambda)}.
\]
Security only for $T=\poly(\lambda)$ can support a polynomial-window reduction, but not the exponential mistake floor.  A sufficient scale is $\varepsilon_\lambda(T)\le 2^{-\eta_\alpha\lambda}$ uniformly for $T\le 2^{\alpha\lambda}$, for some $\alpha>c(d+1)$ and $\eta_\alpha>0$, for all sufficiently large $\lambda$.

Finite-block pseudorandom error-correcting codes have the relevant hidden-codeword flavor \citep{christgunn2024pseudorandom}, but do not supply the required infinite-index common-point property.  An infinite family with uniformly bounded pairwise agreement could instead force finite closure dimension; without \textup{(R2)}--\textup{(R3)}, however, that route would not reproduce the zero-mistake side.  No standard-model family satisfying the checklist is asserted here.

\subsection{Other open problems}

Two further directions remain open.  The first is hardness for a natural class such as regular or context-free languages; current cryptographic hardness for next-symbol prediction of regular languages \citep{bhattamishra2025nsphardness} does not immediately transfer to positive-only generation.  The second is a structural theory of efficient closure generation: the right notion must distinguish an infinite closure from an efficiently searchable one.

The random-oracle theorem supplies a clean test case for all three questions.  For this construction, it pinpoints the unbounded search hidden in the information-theoretic closure argument and quantifies the resulting mistake cost under bounded-query access.

\paragraph{AI Disclosure.} The authors acknowledge the use of GPT-5.6 as an assistive tool in preparing this manuscript, including support in drafting and polishing the exposition of detailed proof steps, generating the cover-page illustration (\cref{fig:separation-at-a-glance}), and preparing several TikZ figures intended to aid readers' understanding.  All AI-assisted material was carefully reviewed, edited, and validated by the authors, who take full responsibility for the final manuscript and its results.

\begingroup
\setlength{\bibsep}{2pt plus .3ex}
\bibliographystyle{alpha-arxiv}
\bibliography{refs}
\endgroup

\clearpage
\appendix
\begin{center}
\Huge{Appendix}
\end{center}

\paragraph{Organization of the appendix.} \Cref{app:oracle-space} formalizes the probability space and measurability conventions. \Cref{app:schedule-proofs} proves the schedule-robust extension, and \cref{app:exponent-necessity} shows why its exponent must depend on the generator. \Cref{app:worked-generation} closes with a concrete execution of the zero-mistake algorithm.

\section{Random-oracle probability space and measurability}
\label[section]{app:oracle-space}

All almost-sure statements use one product space.  Let
\[
  \mc D\defeq
  \bigcup_{\lambda\ge1}\bigl(\{0,1\}^{\lambda}\times\N\bigr),
  \qquad
  \Omega_H\defeq
  \prod_{(b,i)\in\mc D}\{0,1\}^{|b|}.
\]
We equip each finite factor with the uniform measure and equip $\Omega_H$ with the resulting product sigma-algebra and product measure.  A point $\omega=(\omega_{b,i})_{(b,i)\in\mc D}$ defines
\[
  H_\omega(b,i)\defeq\omega_{b,i}.
\]
Thus one oracle is sampled simultaneously at every seed length.  If $(b_1,i_1),\ldots,(b_k,i_k)$ are distinct inputs and $y_r\in\{0,1\}^{|b_r|}$, then
\[
  \P_H\!\left[H(b_r,i_r)=y_r\text{ for every }1\le r\le k\right]
  =2^{-\sum_{r=1}^k|b_r|}.
\]
This coordinate identity supplies every oracle-coordinate independence assertion used in \cref{lem:construction,lem:s3-formal}.

\paragraph{Measurability.} A finite oracle transcript determines a cylinder event.  An event specified by a halting finite-query execution is a countable union of such transcript cylinders, and a finite-horizon mistake event is therefore measurable.  The halting event for the unbounded search in \cref{alg:zero-mistake} is likewise the union, over finite scan cutoffs, of cylinder events. More explicitly, the finite-horizon mistake count is a bounded jointly measurable function of $(H,\rho_G)$.  Integrating over $\rho_G$ therefore gives a measurable function $H\mapsto\E_G[\mist_{<W}]$ by Tonelli's theorem.  Hence the threshold events used to define $f_{\lambda,G}$ and $\mc E_G$ are measurable; eventual-in-$\lambda$ statements use only countable unions and intersections of such events. Randomized machines may be supplied with an independent fair-coin product space $\Omega_G=\{0,1\}^{\N}$; expectations denoted by $\E_G$ integrate only over this additional coordinate after the seed and oracle have been fixed.

\section{Schedule-robust finite-horizon reduction}
\label[section]{app:schedule-proofs}

The random-stopping identity below is valid for any random distinct schedule.  Precommitment enters only when the stopped experiment is interpreted as a fresh-point predictor.

\begin{lemma}[Low mistakes yield a fresh prediction]
\label[lemma]{lem:random-stopping}
Fix $\lambda$, a seed--oracle pair $(b,H)$, and an integer $W\ge1$.  For a possibly random distinct index sequence $\Pi=(\Pi_1,\Pi_2,\ldots)$, put
\[
  x_t=E(b,\Pi_t),
  \qquad
  Q_t=\{\Pi_1,\ldots,\Pi_t\},
  \qquad
  Q_0=\varnothing,
\]
and let $\sigma_{b,\Pi}=(x_1,x_2,\ldots)$.  Suppose that, over all randomness governing $G$ and $\Pi$,
\[
  \E_{G,\Pi}\!\left[
    \mist_{<W}(G,\Lb,\sigma_{b,\Pi})
  \right]\le m
\]
for some $0\le m<W$.

Draw $t^*$ uniformly from $\{0,\ldots,W-1\}$, independently of all other randomness, and consider the stopped output $a_{t^*}$ under the ordered prefix $(x_1,\ldots,x_{t^*})$.  Then
\[
  \P_{G,\Pi,t^*}\!\left[
    a_{t^*}=E(b,j)\ \text{for some }j\notin Q_{t^*}
  \right]
  \ge 1-\frac{m}{W}.
\]
This probability statement requires no independence assumption on $\Pi$.  If, in addition, the schedule prefix can be sampled without oracle queries independently of $(b,H)$ and of $G$'s coins, then the stopped experiment is an admissible predictor for \cref{lem:s3-formal}.  Alternatively, after conditioning on a realized precommitted ordered prefix, that finite prefix can be hard-coded into an admissible predictor.  If producing output $a_i$ uses at most $q_G(\lambda,i)$ oracle queries, its simulation of $G$ uses at most
\[
  T_W\defeq\sum_{i=0}^{W-1}q_G(\lambda,i)
\]
oracle queries.
\end{lemma}

\begin{proof}
For $0\le t<W$, let
\[
  Z_t
  \defeq
  \mathbf 1\!\left[
    a_t\in\Lb\setminus\{x_1,\ldots,x_t\}
  \right].
\]
Since every one of the first $W$ outputs is either a mistake or a non-mistake,
\[
  \mist_{<W}
  =W-\sum_{t=0}^{W-1}Z_t.
\]
Consequently,
\begin{align*}
  \P_{G,\Pi,t^*}[Z_{t^*}=1]
  &=\frac1W\sum_{t=0}^{W-1}\E_{G,\Pi}[Z_t]\\
  &=1-\frac1W
    \E_{G,\Pi}\!\left[
      \mist_{<W}(G,\Lb,\sigma_{b,\Pi})
    \right]\\
  &\ge1-\frac{m}{W}.
\end{align*}
On the event $Z_{t^*}=1$, property~\textup{(S1)} implies that $a_{t^*}=E(b,j)$ for a unique index $j$.  Since this output is not among $x_1,\ldots,x_{t^*}$, distinctness of the index sequence gives $j\notin Q_{t^*}$.  Thus every non-mistake selected by the random stopping rule is a correct target point fresh from the revealed sample set.

When the additional independence and sampler conditions in the statement hold, the predictor may reproduce the generator's state at round $t^*$ by simulating outputs $a_0,\ldots,a_{t^*}$.  Hence its oracle-query cost is at most
\[
  \sum_{i=0}^{t^*}q_G(\lambda,i)
  \le \sum_{i=0}^{W-1}q_G(\lambda,i)=T_W.
\]
\end{proof}

\begin{proposition}[Finite-horizon bound for precommitted schedules]
\label[proposition]{prop:precommitted-joint}
Let $\Pi_\lambda$ be any precommitted index schedule, let $W\ge1$, and let $0\le m<W$.  If $G$ uses at most $q_G(\lambda,i)$ oracle queries while producing output $i$, let $T_W$ be as in \cref{lem:random-stopping}.  Then
\[
  \P_{b,H}\!\left[
    \E_{G,\Pi}\!\left[
      \mist_{<W}(G,\Lb,\sigma_{b,\Pi})
    \right]\le m
  \right]
  \le
  \frac{T_W+1}{1-m/W}\,2^{-\lambda}.
\]
For a deterministic schedule, the expectation over $\Pi$ is omitted.
\end{proposition}

\begin{proof}
Let $B_\lambda$ denote the event inside the probability on the left-hand side and set $\alpha=1-m/W>0$.  For every $(b,H)\in B_\lambda$, the random-stopping identity of \cref{lem:random-stopping} gives
\[
  \P[\text{the stopped output is a correct fresh point}\mid b,H]
  \ge\alpha.
\]
Averaging over $(b,H)$ therefore gives
\[
  \P[\text{the stopped output is a correct fresh point}]
  \ge \alpha\,\P_{b,H}[B_\lambda].
\]

For the reverse bound, condition on $t^*$ and on the ordered schedule prefix
\[
  (\Pi_{\lambda,1},\ldots,\Pi_{\lambda,t^*}).
\]
The associated sample-index set
\[
  Q_{t^*}
  =
  \{\Pi_{\lambda,1},\ldots,\Pi_{\lambda,t^*}\}
\]
is then finite and fixed.  Because the schedule and stopping round were chosen independently of $(b,H)$, this conditioning preserves exactly the independence hypothesis of \cref{lem:s3-formal}.  Hard-code the conditioned ordered prefix into a predictor, which reconstructs the ordered history by matching each scheduled index to the literal index field of its supplied sample and makes at most $T_W$ oracle queries.  Hence \cref{lem:s3-formal} bounds its conditional success probability by
\[
  (T_W+1)2^{-\lambda}.
\]
The bound is uniform over the conditioned stopping round and ordered prefix, so averaging over them preserves it.  Combining the lower and upper bounds and dividing by $\alpha$ proves the claim.
\end{proof}

\begin{proof}[Proof of \cref{thm:schedule-robust}]
Since $W_c=2^{c\lambda+O(1)}$,
\[
  T_{W_c}
  \le
  r_G(\lambda)\sum_{i=0}^{W_c-1}(i+1)^d
  \le
  r_G(\lambda)W_c^{d+1}
  =
  2^{c(d+1)\lambda+o(\lambda)}.
\]
Also,
\[
  1-\frac{m_c}{W_c}
  =\frac{m_c+2}{2(m_c+1)}
  >\frac12.
\]
Therefore \cref{prop:precommitted-joint} gives
\[
  \P_{b,H}\!\left[
    \E_{G,\Pi}[\mist_{<W_c}]\le m_c
  \right]
  \le
  2(T_{W_c}+1)2^{-\lambda}
  =
  2^{-[1-c(d+1)]\lambda+o(\lambda)}.
\]
Any fixed $\gamma<1-c(d+1)$ is therefore valid for all sufficiently large $\lambda$.

For a fixed pair $(G,\Pi)$, let
\[
  f_\lambda(H)
  =
  \P_b\!\left[
    \E_{G,\Pi}[\mist_{<W_c}]\le m_c
  \right].
\]
Tonelli's theorem (\cref{fact:tonelli}) and the preceding joint bound give $\E_H[f_\lambda]\le2^{-\gamma\lambda}$.  Markov's inequality (\cref{fact:markov}) yields
\[
  \P_H[f_\lambda\ge\lambda^{-2}]
  \le \lambda^2 2^{-\gamma\lambda}.
\]
The right-hand side is summable, so the first Borel--Cantelli lemma (\cref{fact:first-borel-cantelli}) implies that almost surely $f_\lambda<1$ for every sufficiently large $\lambda$.  Hence some seed lies outside the low-mistake event at each such length.

Uniform oracle machines and uniform scheduler machines form countable families.  The countable full-measure intersection (\cref{fact:countable-intersection}) therefore places the conclusion for all pairs $(G,\Pi)$ on one oracle event.  For each $G$, one may fix any admissible $c_G$, for example $c_G=1/[2(d+1)]$.
\end{proof}

\section{Why the exponent must depend on the generator}
\label[section]{app:exponent-necessity}

The dependence of the exponent on the generator is not merely a consequence of how \cref{prop:perseed-exp} is written.  Under a query bound polynomial in the round index, no common positive exponent can hold for every generator.

\begin{proposition}[No generator-independent exponent]
\label[proposition]{prop:no-universal-exponent}
For every constant $c_0>0$, there is a deterministic uniform oracle generator $G$ with a per-output query bound polynomial in the output index such that, almost surely over $H$, for every sufficiently large $\lambda$ and every $b\in\{0,1\}^{\lambda}$,
\[
  \mist(G,\Lb,\sigma_b)<2^{c_0\lambda}.
\]
Consequently, the constant in the exponential lower bound cannot be chosen uniformly over all polynomial-query generators.
\end{proposition}

\begin{proof}
Choose an integer $d>1/c_0$.  Let $\mc U_\lambda$ be the event that the three-coordinate signatures
\[
  \beta\longmapsto
  \bigl(H(\beta,1),H(\beta,2),H(\beta,3)\bigr),
  \qquad \beta\in\{0,1\}^{\lambda},
\]
are pairwise distinct.  Let $C_\lambda$ count colliding unordered pairs of signatures.  Any fixed pair collides with probability $2^{-3\lambda}$, so Tonelli's theorem (\cref{fact:tonelli}) and Markov's inequality (\cref{fact:markov}) give
\[
  \P_H[\mc U_\lambda^{\mathsf c}]
  =\P_H[C_\lambda\ge1]
  \le\E_H[C_\lambda]
  =\binom{2^\lambda}{2}2^{-3\lambda}
  \le2^{-\lambda-1}.
\]
The first Borel--Cantelli lemma (\cref{fact:first-borel-cantelli}) implies that almost surely $\mc U_\lambda$ holds for every sufficiently large $\lambda$.

Define
\[
  \tau_\lambda
  \defeq
  \min\bigl\{t\ge3:(t+1)^d\ge3\cdot2^\lambda+1\bigr\}.
\]
Before round $\tau_\lambda$, the generator emits a fixed malformed string.  At every $t\ge\tau_\lambda$, it reads the first three values from the canonical history, queries $H(\beta,k)$ for every $\beta\in\{0,1\}^{\lambda}$ and $k\in\{1,2,3\}$, and searches for the unique matching seed $\widehat b$.  If it exists, the generator makes one additional query and outputs
\[
  E(\widehat b,t+1)=\tuple{\lambda,t+1,H(\widehat b,t+1)};
\]
otherwise it emits the malformed string.

This is one uniform machine with query bound
\[
  q_G(\lambda,t)\le(t+1)^d.
\]
Indeed, it makes no queries before $\tau_\lambda$ and at most $3\cdot2^\lambda+1\le(t+1)^d$ queries thereafter.  On $\mc U_\lambda$, the recovered seed is $\widehat b=b$ for every target, and the index $t+1$ has not yet been revealed under $\sigma_b$. Hence every output from round $\tau_\lambda$ onward is a non-mistake, simultaneously for all length-$\lambda$ targets.  Therefore
\[
  \mist(G,\Lb,\sigma_b)
  \le\tau_\lambda
  =2^{\lambda/d+O(1)}
  <2^{c_0\lambda}
\]
for all sufficiently large $\lambda$.  The almost-sure eventual occurrence of $\mc U_\lambda$ completes the proof.
\end{proof}

\Cref{prop:no-universal-exponent} establishes only that the exponent must depend on $G$; it does not claim that the admissible threshold $c<1/(d+1)$ in \cref{prop:perseed-exp} is optimal.

\section{A worked execution of the zero-mistake generator}
\label[section]{app:worked-generation}

Run \cref{alg:zero-mistake} on the oracle window of \cref{ex:oracle-window}, with target $L_{00}$ under its canonical enumeration.  Write
\[
\begin{aligned}
  e_1&\defeq E(00,1)=\tuple{2,1,10},
  &\qquad
  e_2&\defeq E(00,2)=\tuple{2,2,01},\\
  e_3&\defeq E(00,3)=\tuple{2,3,11},
  &
  e_4&\defeq E(00,4)=\tuple{2,4,00}.
\end{aligned}
\]
and $\mathsf S_2=\{00,01,10,11\}$.  On the recurrence event of \cref{lem:construction}, define
\[
  j_\star
  \defeq
  \min\{j\ge5:H(00,j)=H(01,j)\},
  \qquad
  y_\star
  \defeq H(00,j_\star)=H(01,j_\star).
\]
The first five calls are as follows.

\begin{table}[H]
\centering
\footnotesize
\caption{Transcript for target $L_{00}$ using the oracle window of \cref{ex:oracle-window}.}
\label[table]{tab:worked-generation}
{\renewcommand{\arraystretch}{1.00}
\setlength{\tabcolsep}{4pt}
\begin{tabularx}{\textwidth}{@{}
  c
  >{\raggedright\arraybackslash}p{.23\textwidth}
  >{\raggedright\arraybackslash}p{.10\textwidth}
  >{\raggedright\arraybackslash}p{.13\textwidth}
  >{\raggedright\arraybackslash}X@{}}
\toprule
$t$ & Revealed history $F_t$ & Revealed indices $I_t$
& Consistent seeds $B_t$ & Increasing scan and returned output \\
\midrule
$0$
& $\varnothing$
& $\varnothing$
& $\mathsf S_2$
& At $j=1$ the value set is $\{00,10,11\}$, and at $j=2$ it is $\{00,01,10\}$; neither is a
singleton.  At $j=3$ all rows give $11$, so $a_0=e_3$. \\
$1$
& $\{e_1\}$, with $x_1=e_1$
& $\{1\}$
& $\{00,01\}$
& The first eligible index is $j=2$, where both surviving rows give $01$.  Hence $a_1=e_2$. \\
$2$
& $\{e_1,e_2\}$, with $x_2=e_2$
& $\{1,2\}$
& $\{00,01\}$
& The first eligible index is $j=3$, where both surviving rows give $11$.  Hence $a_2=e_3$. \\
$3$
& $\{e_1,e_2,e_3\}$, with $x_3=e_3$
& $\{1,2,3\}$
& $\{00,01\}$
& At $j=4$ the surviving values are $\{00,10\}$, so the scan continues beyond the displayed
window and returns $a_3=\tuple{2,j_\star,y_\star}$. \\
$4$
& $\{e_1,e_2,e_3,e_4\}$, with $x_4=e_4$
& $\{1,2,3,4\}$
& $\{00\}$
& Only the target seed remains.  Thus the first eligible index, $j=5$, automatically has a
singleton value set, and $a_4=E(00,5)=\tuple{2,5,H(00,5)}$. \\
\bottomrule
\end{tabularx}
}
\end{table}

Every displayed output is valid and unrevealed.  The repetition $a_2=a_0=e_3$ is allowed until $e_3$ is revealed; the score in \cref{def:model} forbids stale examples, not repetition of the generator's own outputs.  The finite window determines the scans only through index $4$; existence of $j_\star$ comes from the recurrence proved in \cref{lem:construction}.

\end{document}